\documentclass[a4paper,USenglish]{lipics-v2019}
 
\usepackage{microtype}
\usepackage{colonequals}
\usepackage{bm}
\usepackage{tikz}
\usepackage{xcolor}
\usepackage{amssymb}
\usepackage{tikz-qtree}
\usepackage{tikz}
\usetikzlibrary{automata}
\usepackage{mathtools}
\usepackage{chngcntr}
\usepackage{xspace}
\usepackage{apxproof}

\nolinenumbers

\usepackage{algorithm}
\usepackage[noend]{algpseudocode}

\usetikzlibrary{positioning, fit, calc, shapes, arrows, arrows.meta}
\EventEditors{Pablo Barcelo and Marco Calautti}
\EventNoEds{2}
\EventLongTitle{22nd International Conference on Database Theory (ICDT 2019)}
\EventShortTitle{ICDT 2019}
\EventAcronym{ICDT}
\EventYear{2019}
\EventDate{March 26--28, 2019}
\EventLocation{Lisbon, Portugal}
\EventLogo{}
\SeriesVolume{127}
\ArticleNo{19} 

\title{Constant-Delay Enumeration for Nondeterministic Document Spanners}

\author{Antoine Amarilli}{LTCI, France \and Télécom ParisTech, France \and Université Paris-Saclay, France}{}{https://orcid.org/0000-0002-7977-4441}{}
\author{Pierre Bourhis}{CNRS, CRIStAL UMR 9189, France \and Inria Lille, France}{}{https://orcid.org/0000-0001-5699-0320}{Supported by the DeLTA ANR project (ANR-16-CE40-0007).}
\author{Stefan Mengel}{CNRS, France \and CRIL UMR 8188, Lens, France}{}{https://orcid.org/0000-0003-1386-8784}{}
\author{Matthias Niewerth}{University of Bayreuth, Germany}{}{https://orcid.org/0000-0003-2032-5374}{Supported by grant number MA 4938/4–1 from the Deutsche Forschungsgemeinschaft}
\authorrunning{A.\ Amarilli, P.\ Bourhis, S.\ Mengel and M.\ Niewerth}
\Copyright{Antoine Amarilli, Pierre Bourhis, Stefan Mengel, and Matthias Niewerth}

\relatedversion{A full version of the paper is available at
\url{https://arxiv.org/abs/1807.09320}. A follow-up journal version with experiments is
available at \url{https://arxiv.org/abs/2003.02576}.}


\hypersetup{
    colorlinks,
    linkcolor={red!50!black},
    citecolor={blue!50!black},
    urlcolor={blue!30!black}
}

\theoremstyle{theorem}
\newtheoremrep{theorem}{Theorem}
\newtheoremrep{proposition}[theorem]{Proposition}
\newtheoremrep{lemma}[theorem]{Lemma}
\newtheoremrep{claim}[theorem]{Claim}
\newtheoremrep{observation}[theorem]{Observation}

\counterwithin{theorem}{section}

\makeatletter
\newcommand*{\defeq}{\mathrel{\rlap{%
  \raisebox{0.3ex}{$\m@th\cdot$}}%
  \raisebox{-0.3ex}{$\m@th\cdot$}}%
  =}
\makeatother

\newcommand{\card}[1]{\left|{#1}\right|}

\newcommand{\calK}{\mathcal{K}}

\newcommand{\calM}{\mathcal{M}}

\newcommand{\NN}{\mathbb{N}}

\newcommand{\NP}{\textsc{NP}}

\renewcommand{\phi}{\varphi}

\newcommand{\calA}{\mathcal{A}}

\newcommand{\calV}{\mathcal{V}}

\newcommand{\f}{\mathrm{f}}

\newcommand{\var}{\mathsf{mapping}}
\newcommand{\lvar}{\mathsf{locmark}}

\newcommand{\Reach}{\mathsf{Reach}}

\newcommand{\node}{\mathsf{node}}

\newcommand{\Absize}{B}
\newcommand{\Width}{W}
\newcommand{\Widthfull}{W_{\mathrm{c}}}
\newcommand{\Depth}{D}

\newcommand{\lvl}{\mathsf{level}}
\newcommand{\JL}{\mathsf{JL}}
\newcommand{\rlvl}{\mathsf{Rlevel}}
\newcommand{\JS}{\textsc{Jump}}

\newcommand{\spa}[2]{[#1,#2\rangle}
\newcommand{\open}[1]{#1\,{\vdash}}
\newcommand{\close}[1]{{\dashv}\,#1}

\newcommand{\nextlevel}{\textsc{NextLevel}\xspace}

\newcommand{\kUCQ}{k\mathsf{\text{-}UCQ}}

\begin{document}

\maketitle

\begin{abstract}
  We consider the information extraction framework known as \emph{document spanners}, 
  and study the problem of efficiently computing the
  results of the extraction from an input document, where the extraction
  task is described as a sequential \emph{variable-set automaton} (VA).
  We pose this problem in the
  setting of enumeration algorithms, where we can first run a preprocessing phase
  and must then produce the results with a small delay between any two consecutive
  results. Our goal is to have an algorithm which is tractable in combined
  complexity, i.e., in the sizes of the input document and the VA;
  while ensuring the best possible
  data complexity bounds in the input document size, i.e., constant delay in
  the document size. Several recent works at PODS'18 proposed such algorithms but with linear
  delay in the document size or with an exponential dependency in size of the (generally
  nondeterministic) input VA. 
  In particular, Florenzano et al.\ suggest that our desired runtime guarantees cannot be met
  for general sequential VAs.
  We refute this and show that,
  given a nondeterministic sequential VA and an input
  document, we can enumerate the mappings of the VA on the document with the
  following bounds: the 
  preprocessing is linear in the document size and polynomial in the size of the
  VA, and the delay is independent of the document and polynomial in the size of
  the VA. The resulting algorithm thus achieves tractability in combined complexity and the best
  possible data complexity bounds. Moreover, it is rather easy to describe, in particular for 
  the restricted case of so-called extended VAs.
\end{abstract}

\section{Introduction}

Information extraction from text documents is an important problem in data management.
One approach to this task has recently attracted a lot of attention: it uses \emph{document spanners},
a declarative logic-based approach first implemented by IBM in their 
tool SystemT~\cite{systemT} and whose core semantics have then been formalized in~\cite{FaginKRV15}.
The spanner approach 
uses variants of regular expressions (e.g.~\emph{regex formulas} with
variables), compiles them
to variants of finite automata (e.g., \emph{variable-set
automata}, for short \emph{VAs}),
and evaluates them on the input document to extract the data of interest.
After this extraction phase, algebraic operations like joins, unions and projections 
can be performed.
The formalization of the spanner framework in~\cite{FaginKRV15} has led to a
thorough investigation of its properties by the
theoretical database community \cite{Freydenberger17,FreydenbergerKP18,MaturanaRV18,FreydenbergerH18,FlorenzanoRUVV18}.

We here consider the basic task in the spanner framework of efficiently computing the results
of the extraction, i.e., computing without duplicates all tuples of ranges of the input
document (called \emph{mappings}) that satisfy the conditions described by a VA.
As many algebraic operations 
can also be compiled into VAs~\cite{FreydenbergerKP18},
this task actually solves the whole data extraction problem for so-called \emph{regular spanners}~\cite{FaginKRV15}.
While the extraction task is intractable for general VAs~\cite{Freydenberger17}, 
it is known to be tractable if we impose that the VA is \emph{sequential}~\cite{FreydenbergerKP18,FlorenzanoRUVV18},
which requires that all accepting runs actually describe a 
well-formed mapping; we will make this assumption throughout our work.
Even then, however, it may still be unreasonable in practice to materialize all mappings:
if there are $k$ variables to extract, then mappings are $k$-tuples and there
may be up to $n^k$ mappings on an input document of size~$n$, which is
unrealistic if~$n$ is large.
For this reason, recent works~\cite{MaturanaRV18,FlorenzanoRUVV18,FreydenbergerKP18} 
have studied the extraction task
in the setting of \emph{enumeration algorithms}: instead of
materializing all mappings, we
enumerate them one by one 
while ensuring that the \emph{delay} between two results is always small.
Specifically, \cite[Theorem~3.3]{FreydenbergerKP18} has shown how to enumerate
the mappings with delay linear in the input document and quadratic in the
VA, i.e., given a document~$d$ and a 
functional VA $A$ (a subclass of sequential VAs),
the delay is
$O(\card{A}^2 \times \card{d})$.

Although this result ensures tractability in both the size of the input document and the automaton, 
the delay may still be long as~$\card{d}$ is generally very large. By contrast, enumeration
algorithms for database tasks
often enforce stronger
tractability guarantees in data complexity~\cite{Segoufin14,Wasa16}, in
particular \emph{linear preprocessing} and \emph{constant delay}
(when measuring complexity
in the RAM model with uniform cost measure~\cite{AhoHU74}).
Such algorithms consist of two phases: 
a \emph{preprocessing phase} which precomputes an index data structure in linear
data complexity, %
and an \emph{enumeration phase} which produces all results so that the delay
between any two consecutive results is always \emph{constant}, i.e., independent
from the input data.
It was recently shown in~\cite{FlorenzanoRUVV18} that this strong guarantee could be achieved when
enumerating the mappings of VAs if we only focus on data complexity, i.e., 
for any \emph{fixed} 
VA, we can enumerate its
mappings with linear preprocessing and constant delay in the input document.
However, the preprocessing and delay
in~\cite{FlorenzanoRUVV18}
are exponential in the VA
because they first determinize it~\cite[Propositions~4.1~and~4.3]{FlorenzanoRUVV18}.
This is problematic because the VAs constructed from regex
formulas~\cite{FaginKRV15} are generally nondeterministic.

Thus, to efficiently enumerate the results of the extraction, we would ideally
want to have the best of both worlds:
ensure that the \emph{combined complexity} (in the sequential VA and in the document) remains polynomial,
while ensuring that the \emph{data complexity} (in the document) is as small as
possible, i.e., linear time for the preprocessing phase and constant time for the delay of the enumeration phase. 
However, up to now, there was no known algorithm to satisfy these requirements while working on nondeterministic sequential VAs.
Further, it was conjectured that such an algorithm is 
unlikely to exist~\cite{FlorenzanoRUVV18} because the 
related task of \emph{counting} the number of mappings is \textsc{SpanL}-hard
for such VAs.

The question of nondeterminism is also unsolved for the related problem of enumerating the results of monadic second-order (MSO) queries on words and trees:
there are several approaches for this task where the query is given as an automaton,
but they require the automaton to be deterministic~\cite{bagan2006mso,amarilli2017circuit_extended}
or their delay is not constant in the input document~\cite{losemann2014mso}.
Hence, also in the context of MSO enumeration, it is not known whether we can achieve linear preprocessing and constant delay in data complexity
while remaining tractable in the (generally non-deterministic) automaton. The
result that we will show in the present paper will imply that we can achieve
this for MSO queries on words when all free variables are first-order, with the
query being represented as a generally non-deterministic sequential VA, or as a
sequential regex-formula with capture variables: note that an extension to trees
is investigated in our follow-up work~\cite{amarilli2019enumeration}.

\subparagraph*{Contributions.}
In this work, 
we show that nondeterminism is in fact not an obstacle 
to enumerating the results of document spanners: we present an algorithm that enumerates the mappings of a nondeterministic sequential VA in polynomial combined complexity 
while ensuring linear preprocessing and constant delay in the input document. This answers the open question of~\cite{FlorenzanoRUVV18}, and improves on the bounds of~\cite{FreydenbergerKP18}. More precisely, we show:

\begin{theorem}
  \label{thm:main}
  Let $2 \leq \omega \leq 3$ be an exponent for Boolean matrix multiplication.
  Let $\calA$ be a sequential VA with variable set $\calV$ and with
  state set~$Q$, and let $d$ be an input document. We can
  enumerate the mappings
  of~$\calA$ on~$d$ with preprocessing time in
  $O((\card{Q}^{\omega+1} + \card{\calA}) \times \card{d})$ and with delay
  $O(\card{\calV} \times (\card{Q}^2 + \card{\calA} \times \card{\calV}^2))$, i.e.,
  linear preprocessing and constant delay in the input
  document, and polynomial preprocessing and delay in the input VA.
\end{theorem}

The existence of such an algorithm is surprising but in hindsight not entirely
unexpected: remember that, in formal language theory, when we are given a word
and a nondeterministic finite automaton, then we can evaluate the automaton on
the word with tractable combined complexity by determinizing the automaton ``on
the fly'', i.e., computing at each position of the word the set of states where
the automaton can be.
Our algorithm generalizes this intuition, and extends it to the task of
enumerating mappings without duplicates: we first present it for so-called
\emph{extended sequential VAs}\footnote{Note that, contrary to what the
terminology suggests, VAs are not special cases of extended VAs.
Further, while extended VAs can be converted in PTIME to VAs, the
converse is not true as there are extended VAs for which the smallest
equivalent VA has exponential size~\cite{FlorenzanoRUVV18}.}, a variant of sequential VAs introduced
in~\cite{FlorenzanoRUVV18}, before generalizing it to sequential VAs.
Our overall approach is to construct a kind of product of the input document
with the extended VA, similarly to~\cite{FlorenzanoRUVV18}. We then use several tricks to ensure the constant delay bound despite nondeterminism; in particular we precompute a \emph{jump function} that allows us to skip quickly the parts of the document where no variable can be assigned. The resulting algorithm is rather simple and has no large hidden
constants.
Note that our enumeration algorithm does not contradict the counting
hardness results of~\cite[Theorem~5.2]{FlorenzanoRUVV18}: while our
algorithm \emph{enumerates} mappings with constant delay and without
duplicates, we do not see a way to adapt it to \emph{count} the
mappings efficiently.  This is similar to the enumeration and counting problems
for maximal cliques: we can enumerate maximal cliques with
polynomial delay~\cite{tsukiyama1977new}, but counting them is
\#P-hard~\cite{valiant1979complexity}.

To extend our result to sequential VAs that are not extended, one possibility
would be to convert them to extended VAs, but this necessarily entails an exponential
blowup \cite[Proposition~4.2]{FlorenzanoRUVV18}. We avoid this by adapting our
algorithm to work with non-extended sequential VAs directly.
Our idea for this is to efficiently enumerate 
at each position
the possible sets of markers that can be assigned by the VA: we do so by
enumerating paths in the VA, relying on the fact that the VA is sequential so
these paths are acyclic. The challenge is that the same set of markers can be
captured by many different paths, but we explain how we can explore efficiently
the set of distinct paths with a technique known as \emph{flashlight
search}~\cite{mary2016efficient,read1975bounds}: the key idea is that we can
efficiently determine which partial sets of markers can be extended to the label
of a path (Lemma~\ref{lem:extendToPath}). 

Of course, our main theorem (Theorem~\ref{thm:main}) implies analogous results
for all spanner formalisms that can be translated to sequential VAs.
In particular, spanners are not usually written as automata by users, but
instead given in a form of regular expressions called \emph{regex-formulas},
see~\cite{FaginKRV15} for exact definitions. 
As we can translate sequential regex-formulas to sequential VAs in linear time~\cite{FaginKRV15,FreydenbergerKP18,MaturanaRV18}, our results imply that we
can also evaluate them:

\begin{corollary}
  Let $2 \leq \omega \leq 3$ be an exponent for Boolean matrix multiplication.
  Let $\phi$ be a sequential regex-formula with variable set $\calV$, and let $d$ be an input document. We can
  enumerate the mappings
  of~$\phi$ on~$d$ with preprocessing time in
  $O(\card{\phi}^{\omega+1} \times \card{d})$ and with delay
  $O(\card{\calV} \times (\card{\phi}^2 + \card{\phi} \times \card{\calV}^2))$, i.e.,
  linear preprocessing and constant delay in the input
  document, and polynomial preprocessing and delay in the input regex-formula.
\end{corollary}

Another direct application of our result is for so-called \emph{regular spanners} which are unions of conjunctive queries (UCQs) posed on regex-formulas, i.e., the closure of regex-formulas under union, projection and joins. We again point the reader to~\cite{FaginKRV15,FreydenbergerKP18} for the full definitions. 
As such UCQs can in fact be evaluated by VAs, our result also implies tractability for such representations, as long as we only perform a bounded number of joins:

\begin{corollary}
 For every fixed $k \in \NN$, let $\kUCQ$ denote the class of document spanners
  represented by UCQs over functional regex-formulas with at most $k$
  applications of the join operator. Then the mappings of a spanner in $\kUCQ$
  can be enumerated with linear preprocessing and constant delay in the document
  size, and with polynomial preprocessing and delay in the size of the spanner representation.
\end{corollary}

\subparagraph*{Paper structure.}
In Section~\ref{sec:prelim}, we formally define spanners, VAs, 
and the enumeration problem that we want to solve on them.
In Sections~\ref{sec:extended}--\ref{sec:jump},
we prove our main result (Theorem~\ref{thm:main}) for \emph{extended}
VAs, where the sets of variables that can be assigned at each position are
specified explicitly. We first describe in
Section~\ref{sec:extended} the main part of our preprocessing phase, which
converts the extended VA and input document to a \emph{mapping DAG} whose paths
describe the mappings that we wish to enumerate. We then
describe in Section~\ref{sec:enum} how to enumerate these paths, up to having
precomputed a so-called \emph{jump function} whose computation is explained in 
Section~\ref{sec:jump}.
Last, we adapt our scheme in Section~\ref{sec:flashlight} for
sequential VAs that are not extended.
We conclude in Section~\ref{sec:conc}.

\section{Preliminaries}
\label{sec:prelim}

\subparagraph*{Document spanners.}
We fix a finite alphabet $\Sigma$.
A \emph{document} $d = d_0 \cdots d_{n-1}$ is just a word over~$\Sigma$.
A \emph{span} of $d$ is a pair $\spa{i}{j}$ with $0 \le i \le j \le |d|$ which
represents a substring (contiguous subsequence) of $d$ starting at position $i$ and ending at position $j-1$.
To describe the possible results of an information extraction task, we will use
a finite set $\calV$ of variables, and define a result as a \emph{mapping} from
these variables to spans of the input document. Following~\cite{FlorenzanoRUVV18,MaturanaRV18} but in contrast to~\cite{FaginKRV15}, we will not require mappings to assign all variables: formally, a
\emph{mapping} of~$\mathcal V$ on~$d$ is a function $\mu$ from some domain $\calV'
\subseteq \calV$ to spans of~$d$.
We define a \emph{document spanner} to be a function 
assigning to every input document $d$ a set of
mappings, which denotes the set of results of the extraction task on the document~$d$.

\subparagraph*{Variable-set automata.}
We will represent document spanners using \emph{variable-set automata} (or
\emph{VAs}).
The transitions of a VA can carry 
letters of~$\Sigma$ or
\emph{variable markers}, which are either of the form $\open x$ for a variable $x\in \mathcal V$ (denoting the start of the span assigned to~$x$) or $\close x$ (denoting its end).
Formally, a \emph{variable-set automaton} $\mathcal A$ (or VA) is then defined to be an
automaton $\mathcal A= (Q,q_0, F, \delta)$ where the transition relation
$\delta$ consists of \emph{letter transitions} of the form $(q, a, q')$ for $q,
q'\in Q$ and $a\in \Sigma$, and of \emph{variable transitions} of the form $(q,
\open x, q')$ or $(q, \close x, q')$ for $q, q'\in Q$ and $x\in \mathcal V$. A
\emph{configuration} of a VA is a pair $(q,i)$ where $q\in Q$ and $i$ is a
position of the input document~$d$. A \emph{run} $\sigma$ of $\mathcal A$ on~$d$
is then a sequence  of configurations
\[(q_0,i_0) \xrightarrow{\sigma_1} (q_1, i_1) \xrightarrow{\sigma_2} \cdots
\xrightarrow{\sigma_m} (q_m, i_m)\] where $i_0=0$, $i_m= |d|$, and where for
every $1 \leq j \leq m$:
\begin{itemize}
  \item Either $\sigma_j$ is a letter of~$\Sigma$, we have $i_j = i_{j-1}+1$, we have $d_{i_{j-1}} = \sigma_j$, and $(q_{j-1}, \sigma_j, q_j)$ is a letter transition of~$\calA$;
  \item Or $\sigma_j$ is a variable marker, we have $i_j = i_{j-1}$, and $(q_{j-1}, \sigma_j, q_j)$ is a variable transition of~$\calA$. In this case we say that the variable marker $\sigma_j$ is \emph{read} at position~$i_j$.
\end{itemize}
As usual, we say that a run is \emph{accepting} if $q_m\in F$.
A run is \emph{valid} if it is accepting, every
variable marker is read at most once, if an open marker $\open x$ is read at a
position $i$ then the corresponding close marker $\close x$ is read at a
position $i'$ with $i \le i'$, and if $\open x$ is not read then $\close x$ is not
read either.
From each valid run, we define a mapping where each variable $x \in \calV$ is mapped to the span $\spa{i}{i'}$ such that $\open x$ is read at position~$i$ and $\close x$ is read at position~$i'$; if these markers are not read then $x$ is not assigned by the mapping (i.e., it is not in the domain~$\calV'$).
The \emph{document spanner} of the VA $\calA$ is then 
the function that assigns to every document~$d$ the set of mappings defined by
the valid runs of $\calA$ on $d$: note that the same mapping can be defined by
multiple different runs.
The task studied in this paper is the following:
given a VA $\calA$ and a document $d$, enumerate \emph{without duplicates} the mappings that are assigned
to $d$ by the document spanner of~$\calA$. The enumeration must write each
mapping as a set of pairs $(m, i)$ where $m$ is a variable marker and $i$ is a position of~$d$.

\subparagraph*{Sequential VAs.}
We cannot hope to efficiently enumerate the mappings of arbitrary VAs because it is already 
\NP-complete to decide if,
given a VA $\mathcal A$ and a document $d$, there are any valid runs of~$\mathcal A$
on $d$~\cite{Freydenberger17}. 
For this reason,
we will restrict ourselves to so-called \emph{sequential} VAs~\cite{MaturanaRV18}. A VA $\calA$ is
\emph{sequential} if for every document $d$, every accepting run of~$\calA$ of~$d$ is also valid: this implies that the document spanner of~$\calA$ can simply be defined following the accepting runs of~$\calA$.
If we are given a VA, then we can test in NL whether it is sequential
\cite[Proposition~5.5]{MaturanaRV18}, and otherwise we can convert it to an equivalent sequential VA
(i.e., that defines the same document spanner) with
an unavoidable exponential blowup in the number of variables (not in the number
of states), using existing results:

\begin{proposition}
  \label{prp:makesequential}
  Given a VA $\calA$ on variable set $\calV$, letting $k \colonequals \card{\calV}$ and $r$ be the number of states of~$\calA$, we can compute an equivalent sequential VA $\calA'$ with $3^k r$ states. Conversely, for any $k \in \NN$, there exists a VA $\calA_k$ with 1 state on a variable set with $k$ variables such that any sequential VA equivalent to~$\calA_k$ has at least $3^k$ states.
\end{proposition}

\begin{proof}
  This can be shown exactly like
\cite[Proposition~12]{Freydenberger17} and
\cite[Proposition~3.9]{Freydenberger17b}. In short,
  the upper bound is shown by modifying~$\calA$ to remember in the automaton state which
  variables have been opened or closed, and by re-wiring the
  transitions to ensure that the run is valid: this creates $3^k$ copies of
  every state because each variable can be either unseen, opened, or closed. For the lower bound, 
  \cite[Proposition~3.9]{Freydenberger17b} gives a VA for which any equivalent
  sequential VA must remember the status of all variables in this way.
\end{proof}

All VAs studied in this work will be sequential, and we will further assume
that they are 
\emph{trimmed} in the sense that for every state $q$ there is a document $d$ and
an accepting run of the VA where the state~$q$ appears. This condition can
be enforced in linear time on any sequential VA: we do a graph traversal to identify the
accessible states (the ones that are reachable from the initial
state), we do another graph traversal to identify the co-accessible states (the ones
from which we can reach a final state), and we remove all states that are not accessible or not co-accessible.
We will implicitly assume that all sequential VAs have been trimmed, which 
implies that they cannot contain any cycle of variable
transitions (as such a cycle would otherwise appear in a run, which would not be
valid).

\subparagraph*{Extended VAs.}
We will first prove our results for a variant of sequential VAs
introduced by~\cite{FlorenzanoRUVV18}, called sequential \emph{extended VAs}. An
extended VA on alphabet $\Sigma$ and variable set $\calV$ is an 
automaton $\mathcal A= (Q,q_0, F, \delta)$ where the transition relation
$\delta$ consists of \emph{letter transitions} as before, and of
\emph{extended variable transitions} (or \emph{ev-transitions})
of the form $(q, M, q')$ where $M$ is a possibly empty set of variable markers. Intuitively,
on ev-transitions, the automaton reads multiple markers at once.
Formally, a \emph{run} $\sigma$ of $\mathcal A$ on~$d = d_0 \cdots d_{n-1}$
is a sequence of configurations (defined like before) where letter
transitions and ev-transitions alternate:
\[\hfill
  (q_0,0) \xrightarrow{\smash{M_0}} (q_0', 0) \xrightarrow{\smash{d_0}}
  (q_1,1) \xrightarrow{\smash{M_1}} (q_1', 1) \xrightarrow{\smash{d_1}}
  \cdots 
  \xrightarrow{\smash{d_{n-1}}} (q_n, n) \xrightarrow{\smash{M_n}} (q_n', n)\quad\hfill
\] 
where $(q_i', d_i, q_{i+1})$ is a letter transition of~$\calA$ for all $0
\leq i < n$, and $(q_i, M_i, q'_i)$ is an \mbox{ev-transition} of~$\calA$ for all $0
\leq i \leq n$ where $M_i$ is the set of variable markers \emph{read} at
position~$i$.
Accepting and valid runs are defined like before, and the extended
VA is sequential if all accepting runs are valid, in which case its document spanner is
defined like before.

Our definition of extended VAs is slightly different from~\cite{FlorenzanoRUVV18} because we
allow ev-transitions that read the empty set to change the automaton state. This allows us to make a small additional assumption to simplify our proofs:
we require that the states of extended VAs are partitioned between
\emph{ev-states}, from which only ev-transitions originate
(i.e., the $q_i$ above), and
\emph{letter-states}, from which only letter transitions originate (i.e.,
the $q_i'$ above);
and we impose that the initial state is an ev-state
and the final states are all letter-states. Note that transitions reading the
empty set move from an ev-state to a letter-state, like all other ev-transitions.
Our requirement can be imposed in linear time on any
input extended VA by rewriting each state to one letter-state
and one ev-state, and re-wiring the transitions and changing the
initial/final status of states appropriately. This rewriting preserves
sequentiality and guarantees that any path in the rewritten extended VA must
alternate between letter transitions and ev-transitions.
Hence, we implicitly make this assumption
on all extended VAs from now on.

\begin{example}
  \label{exa:va}
  The top of Figure~\ref{fig:figure} represents a sequential extended VA
  $\calA_0$ to
  extract email addresses. To keep the example readable, we simply define them as words (delimited by a space
  or by the beginning or end of document)
  which contain one at-sign ``\texttt{\textsf{@}}'' preceded and followed by a non-empty sequence
  of non-``\texttt{\textsf{@}}'' characters. 
  In the drawing of~$\calA_0$, the initial state $q_0$ is at the left, and the states
  $q_{10}$ and $q_{12}$ are final. The transitions labeled by $\Sigma$ represent a
  set of transitions for each letter of $\Sigma$, and the same holds for
  $\Sigma'$ which we define as $\Sigma' \colonequals \Sigma \setminus
  \{\text{\texttt{\textsf{@}}}, \text{\textvisiblespace}\}$.
  
  It is easy to see that, on any input document $d$, there is one mapping
  of~$\calA_0$
  on~$d$ per email address contained in~$d$, which assigns the markers $\open x$
  and $\close x$ to the beginning and end of the email address, respectively. In
  particular, $\calA_0$ is sequential, because any accepting run is valid. Note that
  $\calA_0$ happens to have the property that each mapping is produced by exactly one
  accepting run, but our results in this paper do not rely on this property.
\end{example}

\subparagraph*{Matrix multiplication.}
The complexity bottleneck for some of our results will be the complexity of
multiplying two Boolean matrices, which is a long-standing open problem, see
e.g.~\cite{Gall12} for a recent discussion. When stating our results, we will
often denote by $2 \leq \omega \leq 3$ an exponent for Boolean matrix
multiplication: this is a constant such that the product of two $r$-by-$r$
Boolean matrices can be computed in time $O(r^\omega)$. For instance, we can
take $\omega \colonequals 3$ if we use the naive algorithm for Boolean matrix multiplication, and it is obvious that we must have $\omega \geq 2$. The best known upper bound is currently~$\omega < 2.3728639$, see~\cite{Gall14a}.

\section{Computing Mapping DAGs for Extended VAs}
\label{sec:extended}

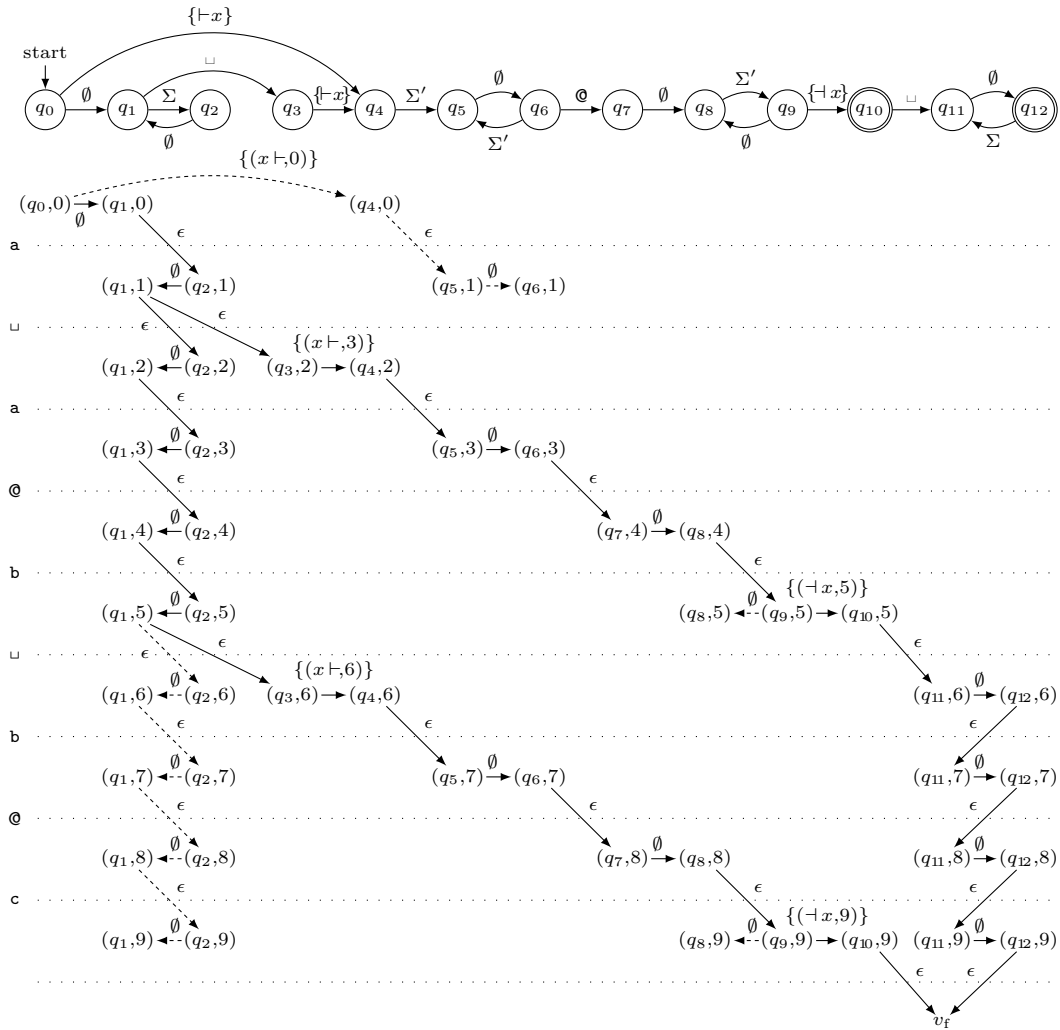
\begin{figure}[t]
  \begin{tikzpicture}[->,>=latex,node distance=10.85mm,auto,
    every state/.style={minimum size=15pt,inner sep=1.5pt},
    map/.style={inner sep=0.5pt},
    trimmed/.style={dash pattern=on 1.5pt off 1.5pt}]
    \scriptsize

    
    \node[state,initial above] (q0) {$q_0$};
    \node[state, right of=q0] (q1) {$q_1$};
    \node[state, right of=q1] (q2) {$q_2$};
    \node[state, right of=q2] (q3) {$q_3$};
    \node[state, right of=q3] (q4) {$q_4$};
    \node[state, right of=q4] (q5) {$q_5$};
    \node[state, right of=q5] (q6) {$q_6$};
    \node[state, right of=q6] (q7) {$q_7$};
    \node[state, right of=q7] (q8) {$q_8$};
    \node[state, right of=q8] (q9) {$q_{9}$};
    \node[state, accepting, right of=q9] (q10) {$q_{10}$};
    \node[state, right of=q10] (q11) {$q_{11}$};
    \node[state, accepting, right of=q11] (q12) {$q_{12}$};

    \path
        (q0) edge node {$\emptyset$} (q1)
        (q1) edge node {$\Sigma$} (q2)
        (q2) edge[out=210,in=-30] node {$\emptyset$} (q1)
        (q1) edge[out=40,in=140] node {\textvisiblespace} (q3)
        (q0) edge[out=45,in=135] node {$\{\open\! x\}$} (q4)
        (q3) edge node[above=-1pt] {$\!\{\!\open\! x\!\}$} (q4)
        (q4) edge node {$\Sigma'$} (q5)
        (q5) edge[out=30,in=150] node {$\emptyset$} (q6)
        (q6) edge[out=210,in=-30] node {$\Sigma'$} (q5)
        (q6) edge node {\texttt{\textsf{@}}} (q7)
        (q7) edge node {$\emptyset$} (q8)
        (q8) edge[out=30,in=150] node {$\Sigma'$} (q9)
        (q9) edge[out=210,in=-30] node {$\emptyset$} (q8)
        (q9) edge node {$\{\!\close x\!\}$} (q10)
        (q10) edge node {\textvisiblespace} (q11)
        (q11) edge[out=30,in=150] node {$\emptyset$} (q12)
        (q12) edge[out=210,in=-30] node {$\Sigma$} (q11)        
        ;

    \node (d1) at (-.4,-1.8) {\texttt{a}};
    \node[below of=d1] (d2)  {\texttt{\textvisiblespace}};
    \node[below of=d2] (d3)  {\texttt{a}};
    \node[below of=d3] (d4)  {\texttt{\textsf{@}}};
    \node[below of=d4] (d5)  {\texttt{b}};
    \node[below of=d5] (d6)  {\texttt{\textvisiblespace}};
    \node[below of=d6] (d7)  {\texttt{b}};
    \node[below of=d7] (d8)  {\texttt{\textsf{@}}};
    \node[below of=d8] (d9)  {\texttt{c}};
    \node[below of=d9] (d10) {};

    \draw[loosely dotted,-] ($ (d1) + (0.3,0) $) -- ($ (d1) + (13.7,0) $);
    \draw[loosely dotted,-] ($ (d2) + (0.3,0) $) -- ($ (d2) + (13.7,0) $);
    \draw[loosely dotted,-] ($ (d3) + (0.3,0) $) -- ($ (d3) + (13.7,0) $);
    \draw[loosely dotted,-] ($ (d4) + (0.3,0) $) -- ($ (d4) + (13.7,0) $);
    \draw[loosely dotted,-] ($ (d5) + (0.3,0) $) -- ($ (d5) + (13.7,0) $);
    \draw[loosely dotted,-] ($ (d6) + (0.3,0) $) -- ($ (d6) + (13.7,0) $);
    \draw[loosely dotted,-] ($ (d7) + (0.3,0) $) -- ($ (d7) + (13.7,0) $);
    \draw[loosely dotted,-] ($ (d8) + (0.3,0) $) -- ($ (d8) + (13.7,0) $);
    \draw[loosely dotted,-] ($ (d9) + (0.3,0) $) -- ($ (d9) + (13.7,0) $);
    \draw[loosely dotted,-] ($ (d10) + (0.3,0) $) -- ($ (d10) + (13.7,0) $);
    
    \newcommand{\map}[2]{$(q_{#1}$,${#2})$};

    \node[map] (m00) at (0,-1.25) {\map 0 0};
    \node[map,right of=m00] (m10) {\map 1 0};
    \node[map,right of=m10] (m20) {};
    \node[map,below of=m20] (m21) {\map 2 1};
    \node[map,below of=m21] (m22) {\map 2 2};
    \node[map,below of=m22] (m23) {\map 2 3};
    \node[map,below of=m23] (m24) {\map 2 4};
    \node[map,below of=m24] (m25) {\map 2 5};
    \node[map,below of=m25] (m26) {\map 2 6};
    \node[map,below of=m26] (m27) {\map 2 7};
    \node[map,below of=m27] (m28) {\map 2 8};
    \node[map,below of=m28] (m29) {\map 2 9};

    \node[map,right of=m20] (m30) {};
    \node[map,left of=m21] (m11) {\map 1 1};
    \node[map,left of=m22] (m12) {\map 1 2};
    \node[map,left of=m23] (m13) {\map 1 3};
    \node[map,left of=m24] (m14) {\map 1 4};
    \node[map,left of=m25] (m15) {\map 1 5};
    \node[map,left of=m26] (m16) {\map 1 6};
    \node[map,left of=m27] (m17) {\map 1 7};
    \node[map,left of=m28] (m18) {\map 1 8};
    \node[map,left of=m29] (m19) {\map 1 9};
    
    \node[map,right of=m30] (m40) {\map 4 0};
    \node[below of=m40] (m41) {};
    \node[map,right of=m41] (m51) {\map 5 1};
    \node[map,right of=m51] (m61) {\map 6 1};

    \node[map,right of=m22] (m32) {\map 3 2};
    \node[map,right of=m32] (m42) {\map 4 2};
    \node[below of=m42] (m43) {};
    \node[map,right of=m43] (m53) {\map 5 3};
    \node[map,right of=m53] (m63) {\map 6 3};
    \node[below of=m63] (m64) {};
    \node[map,right of=m64] (m74) {\map 7 4};
    \node[map,right of=m74] (m84) {\map 8 4};
    \node[map,below of=m84] (m85) {\map 8 5};
    \node[map,right of=m85] (m95) {\map 9 5};
    \node[map,right of=m95] (m105) {\map {1\!0} 5};
    \node[below of=m105] (m106) {};
    \node[map,right of=m106,node distance=9.5mm] (m116) {\map {1\!1} 6};
    \node[map,below of=m116] (m117) {\map {1\!1} 7};
    \node[map,below of=m117] (m118) {\map {1\!1} 8};
    \node[map,below of=m118] (m119) {\map {1\!1} 9};

    \node[map,right of=m116, node distance=11.35mm] (m126) {\map {1\!2} 6};
    \node[map,below of=m126] (m127) {\map {1\!2} 7};
    \node[map,below of=m127] (m128) {\map {1\!2} 8};
    \node[map,below of=m128] (m129) {\map {1\!2} 9};

    \node[map,right of=m26] (m36) {\map 3 6};
    \node[map,right of=m36] (m46) {\map 4 6};
    \node[below of=m46] (m47) {};
    \node[map,right of=m47] (m57) {\map 5 7};
    \node[map,right of=m57] (m67) {\map 6 7};
    \node[below of=m67] (m68) {};
    \node[map,right of=m68] (m78) {\map 7 8};
    \node[map,right of=m78] (m88) {\map 8 8};
    \node[map,below of=m88] (m89) {\map 8 9};
    \node[map,right of=m89] (m99) {\map 9 9};
    \node[map,right of=m99] (m109) {\map {1\!0} 9};

    \node[map,below of=m119] (mf) {$v_{\f}$};

      
    \path (m00) edge node[below,near start] {$\emptyset$} (m10);
    \path (m10) edge node {$\epsilon$} (m21);
    \path (m21) edge node[above,near start] {$\emptyset$} (m11);
    \path (m22) edge node[above,near start] {$\emptyset$} (m12);
    \path (m23) edge node[above,near start] {$\emptyset$} (m13);
    \path (m24) edge node[above,near start] {$\emptyset$} (m14);
    \path (m25) edge node[above,near start] {$\emptyset$} (m15);
    \path (m26) edge[trimmed] node[above,near start] {$\emptyset$} (m16);
    \path (m27) edge[trimmed] node[above,near start] {$\emptyset$} (m17);
    \path (m28) edge[trimmed] node[above,near start] {$\emptyset$} (m18);
    \path (m29) edge[trimmed] node[above,near start] {$\emptyset$} (m19);

    \path (m00) edge[trimmed,out=15,in=165] node[pos=0.58] {$\{(\open x\!,\!0)\}$} (m40);
    \path (m40) edge[trimmed] node {$\epsilon$} (m51);
    \path (m51) edge[trimmed,near start] node {$\emptyset$} (m61);

    \path (m11) edge node[below left,pos=0.3] {$\epsilon$} (m22);
    \path (m12) edge node[above right] {$\epsilon$} (m23);
    \path (m13) edge node[above right] {$\epsilon$} (m24);
    \path (m14) edge node[above right] {$\epsilon$} (m25);
    \path (m15) edge[trimmed] node[below left,pos=0.3] {$\epsilon$} (m26);
    \path (m16) edge[trimmed] node[above right] {$\epsilon$} (m27);
    \path (m17) edge[trimmed] node[above right] {$\epsilon$} (m28);
    \path (m18) edge[trimmed] node[above right] {$\epsilon$} (m29);

    \path (m11) edge node {$\epsilon$} (m32);
    \path (m32) edge node[above=1mm] {$\{(\open x$,$2)\}$} (m42);
    \path (m42) edge node {$\epsilon$} (m53);
    \path (m53) edge node[near start] {$\emptyset$} (m63);
    \path (m63) edge node {$\epsilon$} (m74);
    \path (m74) edge node[near start] {$\emptyset$} (m84);
    \path (m84) edge node {$\epsilon$} (m95);
    \path (m95) edge[trimmed] node[above,near start] {$\emptyset$} (m85);
    \path (m95) edge node[above=1mm] {$\{(\close x$,$5)\}$} (m105);
    \path (m105) edge node {$\epsilon$} (m116);
    \path (m116) edge node[near start] {$\emptyset$} (m126);
    \path (m126) edge node[above left] {$\epsilon$} (m117);
    \path (m117) edge node[near start] {$\emptyset$} (m127);
    \path (m127) edge node[above left] {$\epsilon$} (m118);
    \path (m118) edge node[near start] {$\emptyset$} (m128);
    \path (m128) edge node[above left] {$\epsilon$} (m119);
    \path (m119) edge node[near start] {$\emptyset$} (m129);
    \path (m129) edge node[above left] {$\epsilon$} (mf);
    
    \path (m15) edge node {$\epsilon$} (m36);
    \path (m36) edge node[above=1mm] {$\{(\open x\!,\!6)\}$} (m46);
    \path (m46) edge node {$\epsilon$} (m57);
    \path (m57) edge node[near start] {$\emptyset$} (m67);
    \path (m67) edge node {$\epsilon$} (m78);
    \path (m78) edge node[near start] {$\emptyset$} (m88);
    \path (m88) edge node {$\epsilon$} (m99);
    \path (m99) edge[trimmed] node[above,near start] {$\emptyset$} (m89);
    \path (m99) edge node[above=1mm] {$\{(\close x$,$9)\}$} (m109);
    \path (m109) edge node {$\epsilon$} (mf);
  \end{tikzpicture}
  \caption{Example sequential extended VA $\calA_0$ to extract e-mail addresses (see
  Example~\ref{exa:va}) and example mapping DAG on an example document (see
  Examples~\ref{exa:mapping}, \ref{exa:capture}, \ref{exa:trim}, and \ref{exa:level}).}
  \label{fig:figure}
\end{figure}
We start our paper by studying extended VAs, which are easier to work with
because the set of markers that can be assigned at every position is explicitly
written as the label of a single transition. We accordingly show Theorem~\ref{thm:main} for the case of extended VAs
in Sections~\ref{sec:extended}--\ref{sec:jump}. We will then cover the case of
non-extended VAs in Section~\ref{sec:flashlight}.

To show Theorem~\ref{thm:main} for extended VAs, we will reduce the problem of enumerating the mappings captured
by $\calA$ to that of enumerating path labels in a special kind of directed
acyclic graph (DAG), called a \emph{mapping DAG}. This DAG is intuitively a variant of the product of~$\calA$ and of the document~$d$,
where we represent simultaneously the position in the document and the
corresponding state of~$\calA$. We will no longer care in the mapping DAG about
the labels of letter transitions, so we will erase these labels and call these
transitions \emph{$\epsilon$-transitions}. As for the ev-transitions, we will
extend their labels to indicate the position in the document in addition to the variable markers.
We first give the general definition of a mapping DAG:

\begin{definition}
  A \emph{mapping DAG} consists of a set $V$ of \emph{vertices}, an
  \emph{initial vertex} $v_0 \in V$, a \emph{final vertex} $v_{\f} \in V$, and a set of \emph{edges} $E$ where each edge
  $(s, x, t)$ has a \emph{source vertex} $s \in V$, a \emph{target vertex} $t
  \in V$, and a \emph{label} $x$ that may be $\epsilon$ (in which case we call the edge an
  \emph{$\epsilon$-edge}) or a finite (possibly empty) set of pairs $(m,i)$, where $m$ is a variable marker and $i$ is a position. These edges are called \emph{marker edges}.
  We require that the graph $(V, E)$ is acyclic.
  We say that a mapping DAG is
  \emph{normalized} if every path from the initial vertex to the final vertex
  starts with a marker edge, ends with an $\epsilon$-edge, and
  alternates between marker edges and $\epsilon$-edges.
  
  The \emph{mapping} $\mu(\pi)$ of a path $\pi$ in the mapping DAG is the union of 
  labels of the marker edges of $\pi$: we require of any mapping DAG that, for every path
  $\pi$, this union is
  disjoint. Given a set $U$ of vertices of~$G$, we write
  $\calM(U)$ for the set of mappings of paths from a vertex of~$U$ to the final
  vertex; note that the same mapping may be captured by multiple different
  paths.
  The set of mappings \emph{captured} by~$G$ is then
  $\calM(G) \colonequals \calM(\{v_0\})$.
\end{definition}

Intuitively, the $\epsilon$-edges will correspond to letter transitions
of~$\calA$ (with the letter being erased, i.e., replaced by~$\epsilon$),
and marker edges will correspond to ev-transitions: their labels
are a possibly empty finite set of pairs of a variable marker and position,
describing which variables have been assigned during the
transition. We now explain how we construct a DAG from~$\calA$ and from
a document~$d$, which we call the \emph{product DAG} of~$\calA$ and~$d$, and
which we will show to be a mapping DAG:

\begin{definition}
  Let $\calA = (Q, q_0, F, \delta)$ be a sequential extended VA and
  let $d = d_0 \cdots d_{n-1}$ be an input document.
  The \emph{product DAG} of $\calA$ and $d$ is the DAG
  whose vertex set is $Q \times \{0, \ldots, n\} \cup \{v_{\f}\}$ with
  $v_{\f} \colonequals (\bullet, n+1)$ for some fresh value~$\bullet$.
  Its edges are:
  \begin{itemize}
    \item For every letter-transition $(q, a, q')$ in~$\delta$, for every $0
      \leq i < \card{d}$ such that $d_i = a$, there is an $\epsilon$-edge from
      $(q, i)$ to $(q', i+1)$;
    \item For every ev-transition $(q, M, q')$ in~$\delta$, for
      every $0 \leq i \leq \card{d}$, there is a marker edge from $(q, i)$ to~$(q', i)$ labeled
      with the (possibly empty)
      set $\{(m,i) \mid m \in M\}$.
    \item For every final state $q \in F$, an $\epsilon$-edge from $(q, n)$
      to~$v_{\f}$.
  \end{itemize}
  The initial vertex of the product DAG is $(q_0, 0)$ and the final
  vertex is $v_{\f}$.
\end{definition}

Note that, contrary to~\cite{FlorenzanoRUVV18}, we
do not contract the $\epsilon$-edges but keep them throughout our algorithm.

\begin{example}
  \label{exa:mapping}
  The product DAG of our example sequential extended VA $\calA_0$ and of the
  example document
  $\text{\texttt{a\textvisiblespace{}a\textsf{@}b\textvisiblespace{}b\textsf{@}c}}$ 
  is shown on Figure~\ref{fig:figure}, with the document being written at the left from
  top to bottom. The initial vertex of the DAG is
  $(q_0, 0)$ at the top left and its final vertex is $v_{\f}$ at the bottom. We draw marker edges horizontally, and
  $\epsilon$-edges diagonally. To simplify the example, we only draw the parts
  of the DAG that are reachable from the initial vertex. Edges are
  dashed when they cannot be used to reach the final vertex.
\end{example}

It is easy to see that this construction satisfies the definition:

\begin{claimrep}
  \label{clm:dagwellformed}
  The product DAG of~$\calA$ and~$d$ is a normalized mapping DAG.
\end{claimrep}

\begin{proofsketch}
  The mapping DAG is acyclic and normalized because its edges follow the
  transitions of the extended VA, which we had preprocessed to distinguish
  letter-states and ev-states. Paths in the mapping DAG cannot contain multiple
  occurrences of the same label, because the labels in the mapping DAG include the position in the
  document.
\end{proofsketch}

\begin{proof}
It is immediate that the product DAG is indeed acyclic, because the second
component is
always nondecreasing, and an edge where the second component does not
increase (corresponding to an ev-transition of the VA) must be followed by an edge where
it does (corresponding to a letter-transition of the VA). What is more, we claim that
no path in the product DAG can include two edges whose labels contain the same
pair $(m, i)$, so that the unions used to define the mappings of the mapping DAG
are indeed disjoint. To see this, consider a path from an edge $((q_1, i_1), M_1, (q'_1, i_1))$ to an edge
$((q_2, i_2), M_2, (q'_2, i_2))$ where $M_1 \neq \epsilon$ and $M_2 \neq
\epsilon$, we have 
$i_1 < i_2$ and~$M_1$
and~$M_2$ are disjoint because all elements of~$M_1$ have~$i_1$ as their first
component, and all elements
of~$M_2$ have~$i_2$ as their first component.
Further, the product DAG is also normalized because
$\calA$ is an extended VA that we have preprocessed to distinguish
letter-states and ev-states.
\end{proof}

Further, the product DAG clearly captures what
we want to enumerate. Formally:

\begin{claimrep}
  \label{clm:dagcorrect}
  The set of mappings of $\calA$ on $d$ is exactly the set of
  mappings $\calM(G)$ captured by the product DAG~$G$.
\end{claimrep}

\begin{proof}
  This is immediate as there is a clear bijection between accepting runs
  of~$\calA$ on~$d$ and paths from the initial vertex of~$G$ to its final
  vertex, and this bijection ensures that the label of the path in~$G$ is the
  mapping corresponding to that accepting run.
\end{proof}

\begin{example}
  \label{exa:capture}
  The set of mappings captured by the example product DAG on
  Figure~\ref{fig:figure} is $\{\{(\open x, 2), (\close x, 5)\}, \{(\open x, 6),
  (\close x, 9)\}\}$, and this is indeed the 
  set of mappings of the example extended VA~$\calA_0$ on the example document.
\end{example}

Our task is to enumerate $\calM(G)$ \emph{without
duplicates}, and this is still non-obvious:
because of nondeterminism, the same mapping in the
product DAG may be witnessed by exponentially many paths, corresponding to
exponentially many runs of the nondeterministic extended VA~$\calA$. 
We will present in the next section our algorithm to perform this task on the
product DAG~$G$. To do this, we
will need to preprocess $G$ by \emph{trimming} it, and 
introduce the notion of \emph{levels} to reason about its
structure.

First, we present how to \emph{trim}~$G$. We say that $G$ is \emph{trimmed}
if every vertex~$v$ is both \emph{accessible} (there
is a path from the initial vertex to~$v$) and \emph{co-accessible} (there is a
path from~$v$ to the final vertex). Given a mapping DAG, we can clearly
trim in linear time by two linear-time graph traversals.
Hence, we will always implicitly assume
that the mapping DAG is trimmed. 
If the mapping DAG may be empty once
trimmed, then there are no mappings to enumerate, so our task is trivial.
Hence, we assume in the sequel that the mapping DAG is non-empty after
trimming. Further, if $\calV = \emptyset$ then the only possible mapping is the
empty mapping and we can produce it at that stage, so in the sequel we assume
that~$\calV$ is non-empty.

\begin{example}
  \label{exa:trim}
  For the mapping DAG of Figure~\ref{fig:figure}, trimming eliminates the
  non-accessible vertices (which are not
  depicted) and the non-co-accessible vertices (i.e., those with incoming dashed edges).
\end{example}

Second, we present an invariant on the structure
of~$G$ by introducing
the notion of \emph{levels}:

\begin{definition}
  A mapping DAG $G$ is \emph{leveled} if its vertices $v = (q,
  i)$ are pairs whose second component~$i$ is a nonnegative integer called the
  \emph{level} of the vertex and written $\lvl(v)$,
  and where the following conditions hold:
  \begin{itemize}
    \item For the initial vertex $v_0$ (which has no incoming edges),
      the level is~$0$;
    \item For every $\epsilon$-edge from~$u$ to~$v$, we have $\lvl(v) = \lvl(u)+1$;
    \item For every  marker edge from~$u$ to~$v$, we have $\lvl(v) = \lvl(u)$. Furthermore, all pairs $(m,i)$ in the label of the edge have $i=\lvl(v)$.
  \end{itemize}
  The \emph{depth} $\Depth$ of~$G$ is the maximal level.
  The \emph{width} $\Width$ of~$G$ is the maximal number of vertices that have
  the same level.
\end{definition}

The following is then immediate by construction:

\begin{claimrep}
  \label{clm:leveled}
  The product DAG of~$\calA$ and~$d$ is leveled, and we have $\Width \leq
  \card{Q}$ and $\Depth = \card{d} + 1$.
\end{claimrep}

\begin{proof}
  It is clear by construction that the product DAG satisfies the first three
  points in the definition of a leveled mapping DAG. To see why the last
  point holds, observe that for every edge of the product DAG, for every pair
  $(m,i)$ that occurs in the label of that edge, the second component~$i$ of the
  pair indicates how many letters of~$d$ have been read so far, so the source
  vertex must have level~$i$.

  To see why the width and depth bounds hold, observe that each
  level of the product DAG corresponds to a copy of~$\calA$, so it has at
  most~$\card{Q}$ vertices; and that the number of levels corresponds to the number
  of letters of the document, plus one level for the final vertex.
\end{proof}

\begin{example}
  \label{exa:level}
  The example mapping DAG on Figure~\ref{fig:figure} is leveled, and
  the levels are represented as horizontal layers separated by dotted
  lines: the topmost level is level~0 and the bottommost level is level~10.
\end{example}

In addition to levels, we will need the notion of a \emph{level set}:
  
\begin{definition}
  \label{def:levelset}
  A \emph{level set} $\Lambda$ is a non-empty set of vertices in a leveled
  normalized
  mapping DAG that all have the same level (written $\lvl(\Lambda)$) and which
  are all the source of some marker edge. The singleton $\{v_{\f}\}$ of
  the final vertex is also considered as a level set.
\end{definition}

In particular, letting $v_0$ be the initial vertex, the singleton $\{v_0\}$ is a
level set. Further, if we consider a level set $\Lambda$ which is not the final
vertex, then we can follow
marker edges from all vertices of~$\Lambda$ (and only such edges) to get
to other vertices, and follow $\epsilon$-edges from these vertices (and
only such edges) to get to a new level set~$\Lambda'$ with $\lvl(\Lambda') =
\lvl(\Lambda)+1$.

\section{Enumeration for Mapping DAGs}
\label{sec:enum}

In the previous section, we have reduced our enumeration problem for extended
VAs on documents to an enumeration problem on normalized leveled mapping
DAGs. In this
section, we describe our main enumeration algorithm on such
DAGs and show the following:

\begin{theorem}
  \label{thm:dag}
  Let $2 \leq \omega \leq 3$ be an exponent for Boolean matrix multiplication.
  Given a normalized leveled mapping DAG $G$ of depth $\Depth$ and width
  $\Width$,
  we can enumerate $\calM(G)$ (without duplicates) 
  with preprocessing $O(\card{G} + \Depth \times \Width^{\omega+1})$ and delay
  $O(\Width^{2} \times (r+1))$
  where $r$ is the size of each produced mapping.
\end{theorem}

Remember that, as part of our preprocessing, we have ensured that the
leveled normalized mapping DAG~$G$ has been trimmed. We will
also preprocess~$G$ to ensure that, given any vertex, we can access its adjacency
list (i.e., the list of
its outgoing edges) in some sorted order on the labels, where we assume that
$\emptyset$-edges come last. 
This sorting can be
done in linear time on the RAM model \cite[Theorem~3.1]{grandjean1996sorting},
so the preprocessing is in~$O(\card{G})$.

Our general enumeration algorithm is then presented as
Algorithm~\ref{alg:main}. We explain the missing pieces next. The function
\textsc{Enum} is initially called with $\Lambda = \{v_0\}$, the level set containing
only the initial vertex, and with $\var$ being the empty set.

\begin{algorithm}
  \caption{Main enumeration algorithm}\label{alg:main}
  \begin{algorithmic}[1]
    \Procedure{enum}{$\Lambda, \var$}
      \State $\Lambda' \colonequals \,$\Call{Jump}{$\Lambda$}
      \If{$\Lambda'$ is the singleton $\{v_{\f}\}$ of the final vertex}\label{line:final}
        \State \Call{Output}{$\var$}
      \Else
        \For{$(\lvar, \Lambda'')$ in \Call{\nextlevel}{$\Lambda'$}}
          \State \Call{enum}{$\Lambda'', \lvar \cup \var$}
        \EndFor
      \EndIf
    \EndProcedure
  \end{algorithmic}
\end{algorithm}

For simplicity, let us assume for now that the \textsc{Jump} function just
computes the identity, i.e., $\Lambda'\colonequals\Lambda$. As for the call
$\nextlevel(\Lambda')$, it returns the pairs $(\lvar, \Lambda'')$ where:
\begin{itemize}
  \item The label set $\lvar$ is an edge label
    such that there is a marker edge labeled with $\lvar$ that starts at some vertex
    of~$\Lambda'$
  \item The level set~$\Lambda''$ is formed
    of all the vertices $w$ at level $\lvl(\Lambda')+1$ that can be
reached from such an edge followed by an $\epsilon$-edge. Formally, a vertex~$w$
    is in~$\Lambda''$ if and only if 
    there is an edge labeled $\lvar$ from some
    vertex $v \in \Lambda'$ to some vertex~$v'$, and there is an $\epsilon$-edge
    from~$v'$ to~$w$.
  \end{itemize}
Remember that, as the mapping DAG is normalized, we know that all edges starting at
vertices of the level set~$\Lambda'$ are marker edges (several of which
may have the same label); and for any target~$v'$ of these edges, all edges that
leave~$v'$ are $\epsilon$-edges whose targets~$w$ are at the level
$\lvl(\Lambda')+1$.

It is easy to see that the $\nextlevel$ function can be computed efficiently:

\begin{proposition}\label{prp:extendedVA}
  Given a leveled trimmed normalized mapping DAG $G$ with width $\Width$, and a 
  level set~$\Lambda'$,
  we can enumerate without duplicates
  all the pairs $(\lvar, \Lambda'') \in \nextlevel(\Lambda')$  with delay
  $O(\Width^2 \times \card{\lvar})$ in an order such that $\lvar = \emptyset$ comes
  last if it is returned.
\end{proposition}

\begin{proof}
  We simultaneously go over the sorted lists of the outgoing edges of each
  vertex of~$\Lambda'$, of which there are at most~$\Width$, and we merge them.
  Specifically, as long as we are not done traversing all lists,
  we consider the smallest value of $\lvar$ (according to the order) that occurs at the
  current position of one of the lists. Then, we move forward in each list until
  the list is empty or the edge label at the current position is no longer equal 
  to~$\lvar$, and
  we consider the set~$\Lambda'_2$ of all vertices~$v'$ that are the targets of
  the edges that we have seen. This considers
  at most~$\Width^2$ edges and reaches at most $\Width$ vertices
  (which are at the same level as~$\Lambda'$), and the total time spent reading
  edge labels is in $O(\card{\lvar})$, so the process is
  in~$O(\Width^2 \times \card{\lvar})$ so far. Now, we consider the outgoing
  edges of all vertices~$v' \in \Lambda'_2$ (all are $\epsilon$-edges) and return the
  set~$\Lambda''$ of the vertices~$w$ to which they lead: this only adds 
  $O(\Width^2)$ to the running time because we consider at most $\Width$
  vertices~$v'$ with 
  at most~$\Width$ outgoing edges each. Last, $\lvar = \emptyset$ comes
  last because of our assumption on the order of adjacency lists.
\end{proof}

The design of Algorithm~\ref{alg:main} is justified by the fact that, for
any level set~$\Lambda'$, the set $\calM(\Lambda')$ can be partitioned based on the value
of~$\lvar$. Formally:

\begin{claimrep}
  \label{clm:decompose}
  For any level set $\Lambda$ of~$G$ which is not the final vertex, we have:
  \begin{equation}
    \calM(\Lambda) \quad=\bigcup_{(\lvar, \Lambda'') \in \nextlevel(\Lambda)}
    \{\lvar \cup \alpha \mid \alpha \in \calM( \Lambda'')\}\;. \label{eq:partitioncalL}
    \end{equation}
    Furthermore, this union is disjoint, non-empty, and none of its terms is
    empty.
\end{claimrep}
\begin{proof}
  The definition of a level set and of a normalized mapping DAG ensures that we can decompose any path $\pi$
  from~$\Lambda$
  to~$v_{\f}$ as a marker edge~$e$ from $\Lambda$ to some vertex
  $v'$, an $\epsilon$-edge from~$v'$ to some vertex~$w$, 
  and a path $\pi'$
  from~$w$ to~$v_{\f}$. Further, the set of such $w$ is clearly a level set. Hence,
  the left-hand side of 
  Equation~\eqref{eq:partitioncalL-apx} is included in the right-hand side.
  Conversely, given such $v$, $v'$, $w$, and $\pi'$, we can combine them into a path
  $\pi$, so the right-hand side is included in the left-hand side. This proves
  Equation~\eqref{eq:partitioncalL-apx}.

  The fact that the union is disjoint is because, by
  definition of a leveled mapping DAG, the labels of marker edges starting at
  vertices in~$\Lambda$ include the level as the second component of all pairs
  that they contain, so these pairs cannot occur at a different level, i.e., they cannot occur on
  the path~$\pi'$; so the mappings in $\calM(\Lambda)$ are indeed partitioned
  according to their intersection with the set of labels that occur on the
  level $\lvl(\Lambda)$.

  The fact that the union is non-empty is because $\Lambda$ is non-empty and its
  vertices must be co-accessible so they must have some outgoing
  marker edge, which implies that $\nextlevel(\Lambda)$ is non-empty.

  The fact that none of the terms of the union is empty is because, for each 
  $(\lvar,\Lambda'') \in \nextlevel(\Lambda)$, we know that $\Lambda''$
  is non-empty because the mapping DAG is trimmed so all vertices are
  co-accessible.
\end{proof}

Thanks to this claim, we could easily prove by induction
that Algorithm~\ref{alg:main} correctly enumerates
$\calM(G)$ when \textsc{Jump} is the identity function.  
However, this algorithm would not achieve the desired delay bounds: 
indeed, it may be the case that $\nextlevel(\Lambda')$ only contains
$\lvar = \emptyset$, and then the recursive call to \textsc{Enum} would not make progress
in constructing the mapping,
so the delay would not generally be linear in the size of the mapping.
To avoid this issue, we use the \textsc{Jump} function to directly ``jump'' to a place in the mapping DAG where we can read a
label different from~$\emptyset$. Let us first give the relevant definitions:

\begin{definition}
  \label{def:jump}
  Given a level set $\Lambda$ in a leveled mapping DAG $G$, the
  \emph{jump level} $\JL(\Lambda)$ of~$\Lambda$ is the first level
  $j \geq \lvl(\Lambda)$ containing a vertex $v'$ such that some
  $v \in \Lambda$ has a path to~$v'$ and such that $v'$ is either the
  final vertex or has an outgoing edge with a label which is
  $\neq \epsilon$ and $\neq \emptyset$.
  In particular we have $\JL(\Lambda) = \lvl(\Lambda)$ if some
  vertex in~$\Lambda$ already has an outgoing edge with such a label, or if
  $\Lambda$ is the singleton set containing only the final vertex.
  
  The \emph{jump set} of~$\Lambda$ is then
  $\JS(\Lambda) \colonequals \Lambda$ if
  $\JL(\Lambda) = \lvl(\Lambda)$, and otherwise $\JS(\Lambda)$ is 
  formed of all vertices at level~$\JL(\Lambda)$ to
  which some $v \in \Lambda$ have a directed path whose last edge is
  labeled~$\epsilon$. This ensures that
 $\JS(\Lambda)$ is always a level set.
\end{definition}

The definition of $\JS$ ensures that we can jump from $\Lambda$
to~$\JS(\Lambda)$ when
enumerating mappings, and it will not change the result because we only
jump over
$\epsilon$-edges and $\emptyset$-edges:

\begin{claimrep}
  \label{clm:jump}
  For any level set $\Lambda$ of~$G$, we have
  $\calM(\Lambda) = \calM(\JS(\Lambda))$.
\end{claimrep}

\begin{proof}
  As $\JS(\Lambda)$ contains all vertices from level $\JL(\Lambda)$ that can be reached from
  $\Lambda$,
  any path $\pi$ from a vertex $u \in \Lambda$ to the final vertex
  can be decomposed into a path $\pi_{uw}$ from $u$ to a vertex $w \in
  \JS(\Lambda)$  and a path $\pi_{wv}$ from~$w$ to~$v$. By definition of
  $\JS(\Lambda)$, 
  we know that all edges in~$\pi_{uw}$ are labeled with~$\epsilon$
  or~$\emptyset$, so
  $\mu(\pi) = \mu(\pi_{wv})$. Hence, we have $\calM(\Lambda)
  \subseteq \calM(\JS(\Lambda))$.
  
  Conversely, given a path $\pi_{wv}$ from a vertex $w \in \JS(\Lambda)$ to the
  final vertex, the definition of~$\JS(\Lambda)$ ensures that there is a vertex $u \in
  \Lambda$ and a
  path $\pi_{uw}$ from~$u$ to~$w$, which again consists only of
  $\epsilon$-edges or $\emptyset$-edges.
  Hence, letting $\pi$ be the concatenation of $\pi_{uw}$ and
  $\pi_{wv}$, we have $\mu(\pi_{wv}) = \mu(\pi)$ and $\pi$ is a path
  from $\Lambda$ to the final vertex. Thus, we have $\calM(\JS(\Lambda))
  \subseteq \calM(\Lambda)$, concluding the proof.
\end{proof}

Claims \ref{clm:decompose} and~\ref{clm:jump} imply that
Algorithm~\ref{alg:main} is correct with this implementation of~$\JS$:

\begin{propositionrep}
  \label{prp:correct}
  \textsc{Enum}$(\{v_0\}, \emptyset)$ correctly enumerates $\calM(G)$ (without
  duplicates).
\end{propositionrep}

\begin{proof}
 We show the stronger claim that for every level set $\Lambda$, and for
  every set $\var$ of labels, we have that
  \textsc{Enum}$(\Lambda, \var)$ enumerates (without duplicates) the set
  $\var \uplus \calM(\Lambda) \colonequals \{\var \cup \alpha \mid \alpha \in
  \calM(\Lambda)\}$.
  The base case is when $\Lambda$ is the final vertex, and 
  then $\calM(\Lambda) = \{\{\}\}$ and the algorithm correctly returns $\{\var\}$.

  For the induction case, let us consider a level set $\Lambda$ which is not the
  final vertex, and some set $\var$ of labels.
  We let $\Lambda' \colonequals \JS(\Lambda)$, and by Claim~\ref{clm:jump} we
  have that $\calM(\Lambda')=\calM(\Lambda)$. Now 
  we know by
  Claim~\ref{clm:decompose} that $\calM(\Lambda')$ can be written as in Equation~\eqref{eq:partitioncalL}
  and that the union is disjoint;
  the algorithm evaluates this union. So it suffices to show that, for each 
  $(\lvar, \Lambda'') \in \nextlevel(\Lambda')$, the corresponding iteration of the
  \textbf{for} loop enumerates (without duplicates) the set $(\var \cup
  \lvar) \uplus \calM(\Lambda'')$. 
  By induction hypothesis, the call \textsc{Enum}$(\JS(\Lambda'), \var
  \cup \lvar)$ enumerates (without duplicates) the set  $(\var \cup \lvar)
  \uplus
  \calM(\JS(\Lambda''))$.
  So this establishes that the algorithm is correct.
\end{proof}

What is more, Algorithm~\ref{alg:main} now achieves the desired
delay bounds, as we will show. Of course, this relies on the fact that
the $\JS$ function can be efficiently precomputed and
evaluated. We only state this fact for now, and prove it in the next section:

\begin{proposition}
  \label{prp:jump}
  Given a leveled mapping DAG $G$ with width $\Width$ and depth~$\Depth$, we can
  preprocess $G$ in time $O(\Depth \times \Width^{\omega+1})$ such that, given
  any level set $\Lambda$ of~$G$, we can compute the jump set $\JS(\Lambda)$
  of~$\Lambda$ in time
  $O(\Width^{2})$.
\end{proposition}

We can now conclude the proof of Theorem~\ref{thm:dag} by showing 
that the preprocessing and delay bounds are as claimed.
For the
preprocessing, this is clear: we do the preprocessing in $O(\card{G})$ presented at the
beginning of the section (i.e., trimming, and computing the sorted adjacency lists), followed by that of Proposition~\ref{prp:jump}.
For the delay, we claim:

\begin{claimrep}\label{claim:delay}
  Algorithm~\ref{alg:main} has delay $O(\Width^{2} \times (r+1))$, where $r$ is
  the size of the mapping of each produced path.
In particular,
  the delay is independent of the size of~$G$.
\end{claimrep}

\begin{proofsketch}
  The time to call \textsc{Jump} is in $O(\Width^2)$ by
  Proposition~\ref{prp:jump}, and the time spent to move to the next iteration
  of the
  \textbf{for} loop with a label set $\lvar$
  is in time $O(\Width^2 \times \card{\lvar})$ using
  Proposition~\ref{prp:extendedVA}: now the operations in the loop body run in constant time if
  we represent $\var$ as a linked list so that we do not have to copy it when making the
  recursive call. As Proposition~\ref{prp:extendedVA}
  ensures that~$\emptyset$ comes last, when producing the first solution,
  we make at most~$r+1$ calls to produce a solution of size~$r$, and the time is
  in $O(\Width^{2} \times (r+1))$. We adapt this argument to show that 
  each successive solution is also produced within that bound: note that when
  we use $\emptyset$ in the \textbf{for} loop (which does not contribute to~$r$)
  then the next call to \textsc{Enum} either
  reaches the final vertex or uses a non-empty set which contributes to~$r$.
  What is more, as
  $\emptyset$ is considered last, the corresponding call to \textsc{Enum} is
  tail-recursive, so we can ensure that the size of the stack (and hence the
  time to unwind it) stays~$\leq r+1$.
\end{proofsketch}

\begin{proof}
  Let us first bound the delay to produce the first solution.
  When we enter the \textsc{Enum} function, we call the \textsc{Jump} function
  to produce $\Lambda'$ in time $O(\Width^2)$ by Proposition~\ref{prp:jump},
  and either $\Lambda'$ is the final vertex or some vertex
  in~$\Lambda'$ must have an outgoing edge with a label different
  from~$\emptyset$.
  Then we enumerate $\nextlevel(\Lambda')$
  with delay $O(\Width^2 \times \card{\lvar})$ for each~$\lvar$ using Proposition~\ref{prp:extendedVA}.
  Remember that Proposition~\ref{prp:extendedVA}
  ensures that the label $\emptyset$ comes last;
  so by definition of $\JS$ the first value of~$\lvar$
  that we consider is different from~$\emptyset$.
  At each round of the \textbf{for} loop,
  we recurse in constant time: in particular, we
  do not copy $\var$ when writing $\lvar \cup \var$, as we can represent the set
  simply as a
  linked list.
  Eventually, after $r+1$ calls, by definition of a leveled mapping DAG,
  $\Lambda$
  must be the final vertex, and then we output a mapping of size~$r$ in
  time $O(r)$: the
  delay is indeed in $O(\Width^{2} \times (r+1))$ because the sizes of the
  values of~$\lvar$ seen along the path sum up to~$r$, and the unions  of
  $\lvar$ and $\var$ are always disjoint by definition of a mapping DAG.

  Let us now bound the delay to produce the next solution.
  To do so, we will first observe that when enumerating a mapping of
  cardinality~$r$, then the size of the recursion stack is always~$\leq r+1$.
  This is because Proposition~\ref{prp:extendedVA} ensures that the value $\lvar
  = \emptyset$ is always considered last in the \textbf{for} loop on
  $\nextlevel(\Lambda')$. 
  Thanks to this, every call to
  $\textsc{Enum}$ where $\lvar=\emptyset$ is actually a tail
  recursion, and we can avoid putting another call frame on the
  call stack using tail recursion
  elimination. This ensures that each call frame on the stack (except
  possibly the last one) contributes to the
  size of the currently produced mapping, so that indeed when we reach the final
  vertex of~$G$ then the call stack is no greater than the size of the mapping
  that we produce.

  Now, let us use this fact to bound the delay between consecutive solutions.
  When we move from one solution to another, it means that some \textbf{for} loop
  has moved to the next iteration somewhere in the call stack. To identify this,
  we must unwind the stack: when we produce a mapping of size~$r$, we unwind the stack until
  we find the next \textbf{for} loop that can move forward. 
  By our observation on the size of the stack, the unwinding takes time $O(r)$
  with $r$ is the size of the previously
  produced mapping; so we simply account for this unwinding time as part of the
  computation of the previous mapping. Now, 
  to move to the next iteration of the \textbf{for} loop and do the
  computations inside the loop, we spend a delay $O(\Width^{2} \times
  \card{\lvar})$ by
  Proposition~\ref{prp:extendedVA}. Let $r'$ be the current size of 
  $\var$, including the current $\lvar$. The \textbf{for} loop iteration
  finishes with a recursive call to
  \textsc{Enum}, and we can re-apply our argument about the first solution above
  to argue that this call identifies a mapping of some size $r''$ in delay $O(\Width^2
  \times (r'' + 1))$. However, because the argument $\var$ to the recursive call had size~$r'$, the mapping which is enumerated actually has size $r' + r''$
  and it is produced in delay $O(\Width^2 \times(r'' + 1) + r')$. This means
  that the overall delay to produce the next solution is indeed in $O(\Width^2 \times (r + 1))$ where $r$
  is the size of the mapping that is produced, which concludes the proof.
\end{proof}

\subparagraph*{Memory usage.} We briefly discuss the \emph{memory usage} of the
enumeration phase, i.e., the maximal amount of working memory that we need to
keep throughout the enumeration phase, not counting the
precomputation phase. Indeed, in enumeration algorithms the memory usage
can generally grow to be
very large even if one adds only a constant amount of information at every step.
We will show that this does not happen here, and that the memory
usage throughout the enumeration remains polynomial in~$\calA$ and constant in
the input document size.

All our memory usage during enumeration is in the call stack,
and thanks to tail recursion elimination (see the
proof of Claim~\ref{claim:delay}) we know that the stack depth is at most $r+1$,
where $r$ is the size
of the produced mapping as in the statement of 
Theorem~\ref{thm:dag}.
The local space in each stack frame must store $\Lambda'$ and $\Lambda''$, which
have size $O(\Width)$, and the status of the enumeration of \nextlevel in
Proposition~\ref{prp:extendedVA}, i.e., for every vertex $v \in \Lambda'$, the
current position in its adjacency list: this also has total size $O(\Width)$, so
the total memory usage of these structures over the whole stack is in $O((r+1)
\times \Width)$.
Last, we must also store the variables~$\var$
and~$\lvar$, but their total size of the variables $\lvar$
across the stack is clearly~$r$, and the same holds of $\var$ because each occurrence is
stored as a linked list (with a pointer to the previous stack frame). Hence, the
total memory usage is $O((r+1) \times \Width)$,
i.e., $O((\card{\calV}+1) \times \card{Q})$ in terms of the extended VA.

\section{Jump Function}
\label{sec:jump}
The only missing piece in the enumeration scheme of Section~\ref{sec:enum} is
the proof of Proposition~\ref{prp:jump}. We first explain
the preprocessing for the \textsc{Jump} function, and then the
computation scheme.

\subparagraph*{Preprocessing scheme.}
Recall the definition
of the jump level $\JL(\Lambda)$ and jump set 
$\JS(\Lambda)$ of
a level set~$\Lambda$
(Definition~\ref{def:jump}).
We assume that we have precomputed in $O(\card{G})$ the mapping $\lvl$ associating each
vertex~$v$ to its level $\lvl(v)$, as well as, for each level~$i$, the list of
the vertices $v$ such that $\lvl(v) = i$.

The first part of the preprocessing is then to compute, for every
individual vertex $v$, the jump level  $\JL(v) \colonequals \JL(\{v\})$, i.e.,
the minimal level containing a vertex $v'$ such that $v'$ is reachable from~$v$
and $v'$ is
either the final vertex or has an outgoing edge which is neither an
$\epsilon$-edge nor
an $\emptyset$-edge. We claim:

\begin{claimrep}
  \label{clm:efficientJL}
  We can precompute in $O(\Depth \times \Width^2)$ the jump level $\JL(v)$
  of all vertices $v$ of~$G$.
\end{claimrep}

\begin{proofsketch}
  We do the computation along a topological order: 
we have $\JL(v_{\f}) \colonequals \lvl(v_{\f})$ for the final vertex~$v_{\f}$,
we have $\JL(v) \colonequals \lvl(v)$ if $v$ has an outgoing edge which is not
an $\epsilon$-edge or an $\emptyset$-edge,
and otherwise we have $\JL(v) \colonequals \min_{v \rightarrow w} \JL(w)$.
\end{proofsketch}

\begin{proof}
This construction can be performed iteratively from the final vertex $v_{\f}$
to the initial vertex $v_0$: 
we have $\JL(v_{\f}) \colonequals \lvl(v_{\f})$ for the final vertex~$v_{\f}$,
we have $\JL(v) \colonequals \lvl(v)$ if $v$ has an outgoing edge which is not
an $\epsilon$-edge or an $\emptyset$-edge,
and otherwise we have $\JL(v) \colonequals \min_{v \rightarrow w} \JL(w)$.

This computation can be performed along a reverse
topological order, which by~\cite[Section 22.4]{CormenLRS09} takes linear time
in~$G$. However, note that $G$ has at most $\Depth \times \Width$ vertices, and we only
traverse $\epsilon$-edges and $\emptyset$-edges: we just check the existence of
edges with other labels but we do not traverse them. Now, as each vertex has at
  most $\Width$ outgoing edges labeled~$\emptyset$ and at most~$\Width$ outgoing
  edges labeled $\epsilon$, the number of edges in the DAG that we actually
  traverse is only $O(\Depth \times \Width^2)$, which shows our complexity bound
  and concludes the proof.
\end{proof}

The second part of the preprocessing is to compute, for each level~$i$ of~$G$,
the \emph{reachable levels} $\rlvl(i)\colonequals\{\JL(v) \mid \lvl(v) = i\}$,
which we can clearly do in linear time in the number of vertices of~$G$, i.e.,
in $O(\Depth \times \Width)$. Note that the definition clearly ensures that we have
$\card{\rlvl(i)} \leq \Width$.

\begin{example}
  In Figure~\ref{fig:figure}, the jumping level for nodes
  $(q_1, 3)$ and $(q_2,3)$ is 6 and the jumping level for nodes $(q_5,3)$ and
  $(q_6, 3)$ is 5. Hence, the set of reachable levels
  $\rlvl(3)$ for level~3 is $\{5,6\}$.
\end{example}

Last, the third step of the preprocessing is to compute a reachability matrix from each level to its
reachable levels. Specifically, for any two levels $i < j$ of~$G$, let $\Reach(i,
j)$ be the Boolean matrix of size at most $\Width \times \Width$ which
describes, for each $(u, v)$ with $\lvl(u) = i$ and $\lvl(v) = j$, whether there
is a path from~$u$ to~$v$ whose last edge is labeled~$\epsilon$. We can't afford
to compute all these matrices, but we claim that we can efficiently compute
a subset of them, which will be enough for our purposes:

\begin{claimrep}
  \label{clm:efficientmatrix}
  We can precompute in time $O(\Depth \times \Width^{\omega+1})$ the matrices 
  $\Reach(i,j)$ for all pairs of levels $i < j$ such that $j \in \rlvl(i)$.
\end{claimrep}

\begin{proofsketch}
  We compute them in decreasing order on~$i$: the matrix 
  $\Reach(i, i+1)$ can be computed in time $O(\Width \times \Width)$ from the
  edge relation, and matrices $\Reach(i, j)$ with $j > i+1$ can be computed
  in time $O(\Width^\omega)$
  as the product of $\Reach(i, i+1)$ and $\Reach(i+1, j)$: note that $\Reach(i+1, j)$
  has been precomputed because $j \in \rlvl(i)$ easily
  implies that $j \in \rlvl(i+1)$.
\end{proofsketch}

\begin{proof}
  We compute the matrices in decreasing order on~$i$, then for each
fixed~$i$ in arbitrary order on~$j$:
\begin{itemize}
\item if $j=i$, then $\Reach(i,j)$ is the identity matrix;
\item if $j=i+1$, then $\Reach(i,j)$ can be computed from the edge
  relation of $G$ in time $O(\Width \times \Width)$, because it suffices to
    consider the edges labeled~$\emptyset$ and~$\epsilon$ between levels~$i$
    and~$j$;
\item if $j>i+1$, then $\Reach(i,j)$ is the product of $\Reach(i,i+1)$
  and $\Reach(i+1,j)$, which can be computed in time
  $O(\Width^\omega)$.
\end{itemize}  
In the last case, the crucial point is that $\Reach(i+1,j)$ has already been
precomputed, because we are computing $\Reach$ in decreasing order on~$i$, and
  because we must have $j \in \rlvl(i+1)$. Indeed,  if $j \in \rlvl(i)$, then
  there is a vertex $v$ with $\lvl(v) = i$ such that $\JL(v) = j$, and 
  the inductive definition of $\JL$ implies that $v$ has an edge to 
a vertex $w$ such that
  $\lvl(w)=i+1$ and $\JL(v)=\JL(w)=j$, which witnesses that $j \in \rlvl(i+1)$.

The total running time of this scheme is in 
$O(\Depth \times \Width^{\omega+1})$: indeed we consider each of the $\Depth$
  levels of~$G$, we compute at most
$\Width$ matrices for each level of $G$ because we have
$\card{\rlvl(i)} \leq W$ for any~$i$, and each matrix is computed in time at
  most~$O(\Width^{\omega})$.
\end{proof}

\subparagraph*{Evaluation scheme.}
We can now describe our evaluation scheme for the jump function. Given a level
set~$\Lambda$, we wish to compute $\JS(\Lambda)$. Let $i$ be the level
of~$\Lambda$, and let $j$ be $\JL(\Lambda)$ which we compute as $\min_{v \in \Lambda}\JL(v)$. 
If $j = i$, then $\JS(\Lambda) = \Lambda$ and there is nothing to do. Otherwise,
by definition there must be $v \in \Lambda$ such that $\JL(v) = j$, so $v$
witnesses that $j \in \rlvl(i)$, and we know that we have precomputed the matrix
$\Reach(i, j)$. Now $\JS(\Lambda)$ are the vertices at level~$j$ to
which the vertices of~$\Lambda$ (at level~$i$) have a directed path whose last
edge is labeled~$\epsilon$,
which we can simply compute in time~$O(\Width^2)$ by unioning the
lines that correspond to the vertices of~$\Lambda$ in the matrix $\Reach(i, j)$.

This concludes the proof of Proposition~\ref{prp:jump} and completes the
presentation of our scheme to enumerate the set captured by
mapping DAGs (Theorem~\ref{thm:dag}). Together with
Section~\ref{sec:extended}, this proves Theorem~\ref{thm:main} in the case of
extended sequential VAs.

\section{From Extended Sequential VAs to General Sequential VAs}
\label{sec:flashlight}

In this section, we adapt our main result (Theorem~\ref{thm:main}) to work with
sequential non-extended VAs rather than sequential extended VAs.
Remember that we cannot tractably
convert non-extended VAs into extended VAs
\cite[Proposition~4.2]{FlorenzanoRUVV18}, so we must modify 
our construction in Sections~\ref{sec:extended}--\ref{sec:jump} to work with
sequential non-extended VAs directly.
Our general approach will be the same:
compute the mapping DAG and trim it like in
Section~\ref{sec:extended}, then precompute the jump level and jump set
information as in Section~\ref{sec:jump}, and apply the enumeration scheme of
Section~\ref{sec:enum}. 
The difficulty is that non-extended VAs 
may assign multiple markers at the same word position by taking multiple
variable transitions instead of one single ev-transition.
Hence, when enumerating all possible values for $\lvar$ in
Algorithm~\ref{alg:main},
we need to consider all possible sequences of variable transitions. The
challenge is that there
may be many different transition sequences that assign the same set of markers,
which could lead to duplicates in the enumeration.
Thus, our goal will be to design a replacement to
Proposition~\ref{prp:extendedVA} for non-extended VAs, i.e., enumerate possible values for~$\lvar$ at
each level without duplicates. 

We start as in Section~\ref{sec:extended} by computing the product DAG $G$ of
$\calA$ and of the input document $d = d_0 \cdots d_{n-1}$ with 
vertex set $Q\times \{0,\dots,n\} \cup \{v_{\f}\}$ with $v_{\f} \colonequals
(\bullet, n+1)$ for some fresh value~$\bullet$, and with the following edge set:
  \begin{itemize}
    \item For every letter-transition $(q, a, q')$ of~$\calA$, for every $0 \leq
      i < \card{d}$ such that $d_i = a$, there is an $\epsilon$-edge from $(q, i)$ to $(q', i+1)$;
    \item For every variable-transition $(q, m, q')$ of~$\calA$ (where $m$ is a
      marker), for
      every $0 \leq i \leq \card{d}$, there is an edge from $(q, i)$ to~$(q',
      i)$ labeled with~$\{(m, i)\}$.
    \item For every final state $q \in F$, an $\epsilon$-edge from $(q, n)$
      to~$v_{\f}$.
  \end{itemize}
The initial vertex of~$G$ is $(q_0, 0)$ and the final vertex is
$v_{\f}$. Note that the edge labels are
now always singleton sets or~$\epsilon$; in
particular there are no longer any $\emptyset$-edges. 

We can then adapt most of Claim~\ref{clm:dagwellformed}: the product DAG is
acyclic because all letter-transitions make the second component increase, and
because we know that there cannot be a cycle of variable-transitions in the
input sequential VA $\calA$ (remember that we assume VAs to be trimmed).
We can also trim the mapping DAG in linear time as before, and
Claim~\ref{clm:dagcorrect} also adapts to show that the resulting mapping DAG correctly captures the mappings that we wish to enumerate.
Last, as in Claim~\ref{clm:leveled}, the resulting mapping DAG is still leveled,
the depth $\Depth$ (number of levels) is still $\card{d}+1$, and
the width $\Width$ (maximal size of
a level) is still~$\leq \card{Q}$; we will also define the \emph{complete width}
$\Widthfull$ of~$G$ in this section as the maximal size, over all levels $i$, of
the sum of the number of vertices with level~$i$ and of the number of
\emph{edges} with
a source vertex having level~$i$: clearly we have $\Widthfull \leq
\card{\calA}$.
The main change in Section~\ref{sec:extended} is that the mapping DAG is no
longer normalized, i.e., we may follow several marker edges in succession (staying at the same level) or follow several
$\epsilon$-edges in succession (moving to the next level each time).
Because of this, we change
Definition~\ref{def:levelset} and 
redefine \emph{level sets} to mean any non-empty set of vertices that are at the same
level.

We then reuse the enumeration approach of Section~\ref{sec:enum}
and~\ref{sec:jump}. Even though the mapping DAG is no longer normalized, 
it is not hard to see that with our new definition of level sets we can 
reuse the jump function from Section~\ref{sec:jump} as-is, and we can also reuse
the general approach of Algorithm~\ref{alg:main}.
However, to accommodate for the
different structure of the mapping DAG, we will need a new
definition for \nextlevel: instead of following
exactly one marker edge before an $\epsilon$-edge, we want to be
able to follow any (possibly empty) path of marker edges before an
$\epsilon$-edge. We formalize
this notion as an \emph{$S^+$-path}:

\begin{definition}
  \label{def:spath}
  For $S^+$ a set of labels, an \emph{$S^+$-path} in the mapping DAG $G$
  is a path of $\card{S^+}$ edges that includes no $\epsilon$-edges and where the
  labels of the path are exactly the elements of~$S^+$ in some arbitrary order.
  Recall that the definition of a mapping DAG ensures that there can be no
  duplicate labels on the path, and that the start and end vertices of an
  $S^+$-path must have the same level because no $\epsilon$-edge is traversed in
  the path.

  For $\Lambda$ a level set,
  $\nextlevel(\Lambda)$ is the set of all pairs $(S^+, \Lambda'')$ where:
  \begin{itemize}
    \item $S^+$ is a set of labels such that there is an $S^+$-path that
  goes from some vertex~$v$ of~$\Lambda$ to some
  vertex~$v'$ which has an outgoing $\epsilon$-edge;
\item $\Lambda''$ is the level set containing exactly
  the vertices $w$ that are targets of these $\epsilon$-edges, i.e., 
      there is an $S^+$-path from some vertex $v \in
  \Lambda$ to some vertex $v'$, and there is an $\epsilon$-edge from~$v'$ to~$w$.
  \end{itemize}
\end{definition}

Note that these definitions are exactly equivalent to what we would obtain if we
converted $\calA$ to
an extended VA and then used our original construction. This directly
implies that the modified enumeration algorithm is correct (i.e., 
Proposition~\ref{prp:correct} extends). In particular, the
modified algorithm still uses the jump pointers as computed in Section~\ref{sec:jump} to
jump over positions where the only possibility is $S^+ = \emptyset$, i.e.,
positions where the sequential VA make no variable-transitions.
The only thing that remains is to establish the delay bounds, for which we need to
enumerate \nextlevel efficiently without duplicates
(and replace Proposition~\ref{prp:extendedVA}).
To present our method for this, we will introduce
the
\emph{alphabet size} $\Absize$ as the maximal number, over all levels $j$ of the
mapping DAG $G$, of the different labels that can
occur in marker edges between vertices at level~$j$; in our
construction this value is bounded by the number of different markers, i.e., 
$\Absize \leq 2 \card{\calV}$. We can now state the claim:

\begin{theoremrep}\label{thm:generalVA}
  Given a leveled trimmed mapping DAG $G$ with complete width $\Widthfull$ and alphabet
  size $\Absize$, and a 
  level set
  $\Lambda'$,
  we can enumerate  without duplicates
  all the pairs $(S^+, \Lambda'') \in \nextlevel(\Lambda')$ 
  with delay $O(\Widthfull \times \Absize^2)$
  in an order such that
  $S^+ = \emptyset$ comes last if it is returned.
\end{theoremrep}

\begin{proof}
  Clearly if $\Lambda'$ is the singleton level set consisting only of the final
  vertex, then the set to enumerate is empty and there is nothing to do. Hence,
  in the sequel we assume that this is not the case.

  Let $p$ be the level of $\Lambda'$.  We call $\calK$ the set of possible
  labels at level~$p$, with $\card{\calK}$ being no greater than the alphabet
  size $\Absize$ of~$G$. We fix an arbitrary order
  $m_1< m_2 < \cdots < m_r$ on the elements of~$\calK$. Remember that we want to enumerate 
  $\nextlevel(\Lambda')$, i.e., 
  all pairs $(S^+, \Lambda'')$ of a subset
  $S^+$ of~$\calK$ such that there is an $S^+$-path in~$G$ from a vertex in
  $\Lambda'$ to
  a vertex $v'$ (which will be at level~$p$) with an outgoing $\epsilon$-edge;
  and the set $\Lambda''$ of the targets of these $\epsilon$-edges (at level~$p+1$).
  Let us consider the complete
  decision tree $T_\calK$ on~$m_1, \ldots, m_r$: it is a complete binary
  tree of height~$r + 1$, where, for all $1 \leq i \leq r$, every edge
  at height~$i$ is labeled with $+m_i$ if it is a right child edge and
  with $-m_i$ otherwise. For every node $n$ in the tree, we consider
  the path from the root of~$T_\calK$ to~$n$, and call the \emph{positive
    set} $P_n$ of~$n$ the labels $m$ such that $+m$ appears
  in the path, and the \emph{negative set}~$N_n$ of~$n$ the labels
  $m$ such that $-m$ appears in the path: it is immediate
  that for every node~$n$ of~$T_\calK$ the sets $P_n$ and $N_n$ are a partition
  of $\{m_1, \ldots m_j\}$ where
  $0 \leq j \leq r$ is the depth of~$n$ in~$T_\calK$.

  We say that a node $n$ of~$T_\calK$ is \emph{good} if there is some $P_n/N_n$-path
  in~$G$ starting at a vertex of $\Lambda'$ and 
  leading to a vertex which has an outgoing $\epsilon$-edge.
  Our goal of determining $\nextlevel(\Lambda')$ can then be rephrased as finding the set of all
  positive sets $P_n$ for all good leaves~$n$ of~$T_\calK$ (and the
  corresponding level set~$\Lambda''$), because 
  there is a clear one-to-one correspondence that sends each subset
  $S\subseteq\calK$ to a leaf $n$ of~$T_\calK$ such that $P_n = S$.

  Observe now that we can use Lemma~\ref{lem:extendToPath}
  to determine in time $O(\card{\Widthfull}\times \card{\calK})$, given a node~$n$
  of~$T_\calK$, whether it is good or bad: call the procedure
  on the subgraph of $G$ that is induced by level $p$ (it has size $\leq
  \Widthfull$) and with the sets
  $S^+ \colonequals P_n$
  and 
  $S^- \colonequals N_n$,
  then check in~$G$ whether one of the vertices returned by the procedure 
  has an outgoing $\epsilon$-edge.
  A naive solution to find the good leaves would then be to test them one by one using
  Lemma~\ref{lem:extendToPath}; but a more efficient idea is to use 
  the structure of~$T_\calK$ and the following facts:
  \begin{itemize}
    \item \emph{The root of~$T_\calK$ is always good.}
      Indeed, $G$ is trimmed, so we know that any $v \in \Lambda'$ has a path to
      some
      $\epsilon$-edge.
\item \emph{If a node is good then all its ancestors are good.}
  Indeed, if $n'$ is
  an ancestor of~$n$, and there is a $P_n/N_n$-path in~$G$ starting at a vertex
  of~$\Lambda'$, then this path is also a $P_{n'}/N_{n'}$ path, because
  $P_{n'}\subseteq P_n$ and $N_{n'} \subseteq N_n$. 
\item \emph{If a node $n'$ is good, then it must have at least one good descendant
  leaf $n$.} Indeed, taking any $P_{n'}/N_{n'}$-path 
  that witnesses that~$n'$ is good, we can take the leaf $n$ to be such that
      $P_n \supseteq P_{n'}$ is exactly the set of labels that occur on the
      path, so that the same path witnesses that~$n$ is indeed good.
  \end{itemize}
  Our flashlight search algorithm will rely on these facts.
  We explore~$T_\calK$ depth-first, constructing it on-the-fly as we visit it,
  and we use Lemma~\ref{lem:extendToPath}
  to guide our search: 
  at a node~$n$ of~$T_\calK$ (inductively assumed
  to be good), we call Lemma~\ref{lem:extendToPath} on
  its two children to determine which of them are good (from the facts above, at
  least one of them must be), and we explore recursively the first
  good child, and then the second good child if there is one.
  When the two children are good, we first explore the child labeled $+m$ before
  exploring the child labeled~$-m$: this ensures that if the empty set is produced
  as a label set in 
  $\nextlevel(\Lambda')$ then we always enumerate it last, as we should.
  Once we reach a leaf~$n$ (inductively assumed to be good) then we output its
  positive set of labels $P_n$.

  It is clear that the algorithm only enumerates label sets which occur in
  $\nextlevel(\Lambda')$.
  What is
  more, as the set of good nodes is upwards-closed in~$T_\calK$, the
  depth-first exploration visits all good nodes of~$T_\calK$, so it visits
  all good leaves and produces all label sets that should occur in
  $\nextlevel(\Lambda')$.
  Now, the delay is
  bounded by $O(\card{\Widthfull} \times \card{\calK}^2)$: indeed, whenever we
  are exploring at any node~$n$, we know that the next good leaf will
  be reached in at most $2 \card{\calK}$ calls to the procedure of
  Lemma~\ref{lem:extendToPath}, and we know that the subgraph of~$G$ induced by
  level~$p$ has size bounded by the complete width $\Widthfull$ of~$G$ so each call
  takes time $O(\card{\Widthfull} \times
  \card{\calK})$, including the time needed to verify
  if any of the reachable vertices~$v'$ has an outgoing $\epsilon$-edge: this
  establishes the delay bound of $O(\card{\Widthfull} \times B^2)$ that we
  claimed. Last, while doing
  this verification, we can produce the set $\Lambda''$ of the targets of these
  edges in the same time bound. This
  set~$\Lambda''$ is correct because any such vertex $v'$ has an outgoing
  $\epsilon$-edge and there is a $P_n/N_n$-path from some vertex $v \in \Lambda'$
  to~$v'$. Now, as $P_n \cup N_n = \calK$ and the path cannot traverse an $\epsilon$-edge, then these paths are
  actually $P_n$-paths (i.e., they exactly use the labels in~$P_n$),
  so~$\Lambda''$ is indeed the set that we wanted
  to produce according to Definition~\ref{def:spath}.
  This concludes the proof.
\end{proof}

With this runtime, the delay of Theorem~\ref{thm:dag} becomes $O((r+1) \times
(\Width^2 + \Widthfull \times \Absize^2))$, and we know that
$\Widthfull
\leq \card{\calA}$, that $\Width \leq \card{Q}$, that $r \leq \card{\calV}$, and
that $\Absize \leq 2 \card{\calV}$; so this leads to the overall
delay of $O(\card{\calV} \times (\card{Q}^2 + \card{\calA} \times \card{\calV}^2))$
in Theorem~\ref{thm:main}.

The idea to prove Theorem~\ref{thm:generalVA} is to use a general
approach called \emph{flashlight
  search}~\cite{mary2016efficient,read1975bounds}:
  we will use a search tree on the possible sets of labels on~$\calV$ to iteratively
construct the set~$S^+$ that can be assigned at the current position, and we will avoid useless parts of the
search tree by using a lemma to efficiently check if a
partial set of labels can be extended to a solution. To formalize the notion of
extending a partial set, we will need the notion of \emph{$S^+/S^-$-paths}:

\begin{definition}
  For $S^-$ and $S^+$ two disjoint sets of labels,
  an \emph{$S^+/S^-$-path} in the mapping DAG $G$ is a path of edges that
  includes no $\epsilon$-edges, that includes no edges with a label in~$S^-$,
  and where every label of~$S^+$ is seen exactly once along the path.
\end{definition}

Note that, when $S^+
\cup S^-$ contains all labels used in~$G$,
then the notions of $S^+/S^-$-path and $S^+$-path coincide, but if $G$ contains
some labels not in $S^+ \cup S^-$ then an $S^+/S^-$-path is free to use them
or not, whereas an $S^+$-path cannot use them. The key to prove Theorem~\ref{thm:generalVA}
is to efficiently determine if
$S^+/S^-$-paths exist:
we formalize this as a lemma
which we will apply to the mapping DAG $G$
restricted to the current level (in particular removing $\epsilon$-edges):

\begin{lemmarep}\label{lem:extendToPath}
  Let $G$ be a mapping DAG with no $\epsilon$-edges
  and let $V$ be its vertex set. Given
  a non-empty set $\Lambda' \subseteq V$ of vertices of~$G$ and 
  given two disjoint sets of labels $S^+$ and $S^-$, 
  we can compute in time $O(\card{G} \times \card{S^+})$ the set
  $\Lambda'_2 \subseteq V$ of vertices $v$ such that there is an $S^+/S^-$-path
  from one vertex
  of~$\Lambda'$ to~$v$.
\end{lemmarep}

\begin{proofsketch}
  We first delete all edges from~$G$ with a label in~$S^-$, add a fresh
  source vertex~$s_0$, and remove all vertices that are not reachable from~$s_0$.
  We then follow a topological sort of~$G$ to annotate each vertex~$v$
  with
  the maximal set of labels of~$S^+$ that can be seen along paths 
  from~$s_0$ to~$v$: and we use a failure annotation~$\emptyset$ when there are
  two such paths that
  can see two incomparable sets of labels of~$S^+$. Indeed, as we argue, when this happens
  the vertex~$v$ can never be
  part of an~$S^+/S^-$-path because the definition of~$G$ imposes that each edge label
  occurs at most once on any path, so the partial paths from~$s_0$ to~$v$ can
  never be completed with all missing labels from~$S^+$.
  Hence, we can compute our set~$\Lambda'_2$ 
  simply by returning all the vertices annotated by the whole set~$S^+$.
\end{proofsketch}

\begin{proof}
  In a first step, we delete from~$G$ all edges with a label which is in~$S^-$.
  This ensures that no path that we consider contains any label from $S^-$.
  Hence, we can completely ignore~$S^-$ in what follows.

  In a second step, we add a fresh source vertex~$s_0$ and edges with a fresh
  label $l_0$ from~$s_0$ to each vertex in~$\Lambda'$, we add~$l_0$ to~$S^+$, and we set
  $\Lambda' \colonequals \{s\}$. This allows us to assume that the set~$\Lambda'$ is a
  singleton $\{s\}$.

  In a third step, we traverse~$G$ in linear time from $s_0$ with a
  breadth-first search to remove all vertices that are not reachable
  from~$s_0$. Hence, we can now assume that every vertex in~$G$ is
  reachable from~$s_0$; in particular every vertex except $s_0$ has at least one
  predecessor.

  Now, we follow a reverse topological order on~$G$ to give
  a label $\chi(w) \subseteq S^+$ to each vertex $w \in V$ 
  with predecessors $w_1, \ldots, w_p$
  and to give
  a label $\chi(w_i, w) \subseteq S^+$ to each edge $(w_i, w)$ of~$G$,
  as follows:
  \begin{align*}
    \chi(s)\quad&\colonequals\quad \emptyset\\
    \chi(w_i,w)\quad&\colonequals\quad \large(\chi(w_i) \cup \{\mu(w_i,w)\}\large) \cap S^+ \\
    \chi(w)\quad&\colonequals\quad\begin{cases}
      \chi(w_i,w) & \text{if $w$ has a predecessor $w_i$ with } \chi(w_i,w) =
      \bigcup_{1 \leq j \leq p} \chi(w_j,w) \\
      \emptyset & \text{otherwise} \\
    \end{cases}
  \end{align*}

  The topological order can be computed in time $O(\card{G})$ by~\cite[Section
  22.4]{CormenLRS09}, and computing~$\chi$  along this order
  takes time $O(\card{G} \times \card{S^+})$.

  Intuitively, the labels assigned to a vertex $w$ or an edge $(w_i,w)$
  correspond to the subset of labels from $S^+$ that are read on a
  path starting at $s_0$ and using $w$ as the last vertex (resp., $(w_i,w)$ as last
  edge). However, we explicitly label a vertex $w$ with
  $\emptyset$ if there are two paths starting at $s_0$ that have seen a
  different subset of $S^+$ to reach~$w$.
  Indeed, as we know that any label can occur at most once on
  each path, such vertices and edges can never be part of a path that
  contains all labels from $S^+$. We will formalize this intuition below.
  
  We claim that, for every vertex~$v$, there is an $S^+/S^-$-path from~$s_0$ to~$v$
  if and only if $\chi(v)= S^+$.
  First assume that $\chi(v)= S^+$. We construct a path $P$ by going backwards starting from $v$. We initialize the current vertex $w$ to be~$v$. Now, as long as $\chi(w)$ is non-empty, we pick a predecessor $w_i$ with $\chi(w_i,w) = \chi(w)$, and we know that either $\chi(w_i,w) = \chi(w_i)$ or $\chi(w_i,w) = \chi(w_i) \cup \{\mu(w_i, w)\}$ with $\mu(w_i,w) \in S^+$, and then we assign $w_i$ as our current vertex~$w$. We repeat this process until we reach a current vertex~$w_0$ with $\chi(w) = \emptyset$, which must eventually happen: the DAG is acyclic, and all vertices except~$s_0$ must have a predecessor, and we know by definition that $\chi(s) = \emptyset$. As all elements of~$S^+$ were in~$\chi(w)$, they were all witnessed on $P$, so we know that $P$ is an $S^+/S^-$-path from~$w_0$ to~$v$. Now, we know that there is a path $P'$ from~$s_0$ to~$w_0$ thanks to our third preprocessing step, and we know that $P'$ uses no elements from~$S^-$ by our assumption on the DAG; so the concatenation of $P'$ and $P$ is an $S^+/S^-$-path from~$s_0$ to~$v$.

  For the other direction, assume that there is an $S^+/S^-$-path
  $P= v_1, \ldots, v_r$ from~$v_1 = s_0$ to a vertex~$v_r = v$.
  We show by induction that
  $\chi(v_i)$ contains all labels that have been seen so far on the path from $s_0$
  to $v_i$. For $v_1=s_0$ this is true by definition. For $i>1$, we
  claim that $\chi(v_i)=\chi(v_{i-1}, v_i)$. By way of contradiction,
  assume this were not the case. Then there is an $x\in S^+$ that
  appears in $\chi(v_{i-1}', v_i)$ for some predecessor $v_{i-1}' \neq v_{i-1}$
  of~$v_i$, but does not appear in~$\chi(v_{i-1}, v_i)$, so that $x$ does not appear on the path $v_1\ldots v_i$. 
  But then $x$ cannot appear on the path $v_i\ldots v_r$ either: indeed the fact that $x \in \chi(v_{i-1}', v_i)$ clearly means that there is a path from~$s_0$ to~$v_i$ (via~$v_{i-1}'$) where the label~$x$ appears. Now, 
  as $x$ can occur only once on every path of~$G$, it cannot also appear on the path $v_i\ldots v_r$ that starts at~$v_i$.
  Hence, $x$ does not appear
  in the path 
  $P$ at all, which contradicts the fact that $P$ is an~$S^+/S^-$-path.
  Thus we have indeed
  $\chi(v_i)=\chi(v_{i-1}, v_i)$. But then, since all elements of
  $S^+$ appear on edges in $P$ and are thus added iteratively in the
  construction of the $\chi(v_i)$, we have $S^+=\chi(v_r)= \chi(v)$ as
  desired.

  Hence, once we have computed the labeling~$\chi$, we can compute in
  time $O(\card{G} \times \card{S^+})$ the set~$\Lambda'_2$ by simply
  finding all vertices~$v$ with $\chi(v) = S^+$. This concludes the
  proof.
\end{proof}

We can now use Lemma~\ref{lem:extendToPath}
to prove Theorem~\ref{thm:generalVA}:

\begin{inlineproof}[Proof sketch of Theorem~\ref{thm:generalVA}]
  We restrict our attention to the level $\lvl(\Lambda')$ of the mapping DAG~$G$
  that contains the input level set~$\Lambda'$: in particular we remove
  all $\epsilon$-edges. The resulting mapping DAG has size at most~$\Widthfull$,
  and we call $\calK$ the set
  of labels that it uses, whose cardinality is at most the alphabet size~$\Absize$
  of~$G$. We fix some arbitrary order on~$\calK$.
  Now, let us consider the full decision
  tree $T_\calK$ on~$\calK$ following this order: it is a complete
  binary tree of height $\card{\calK}$, each internal node at depth~$0 \leq
  r < \card{\calK}$ has two children reflecting on whether we take the $r$-th
  label of~$\calK$ or not, and each leaf~$n$ corresponds to a subset of~$\calK$
  built according to the choices described on the path from the root
  of~$T_\calK$ to~$n$. Our algorithm will explore $T_\calK$ to find the
  sets~$S^+$ of labels that we must enumerate for~$\Lambda'$ and~$G$.

  More precisely, we wish to determine the leaves of~$T_\calK$ that correspond to a set~$S^+$
  such that there is an~$S^+$-path in~$G$ from a vertex of~$\Lambda'$ to a vertex with
  an outgoing $\epsilon$-edge: we call this a \emph{good} leaf. The naive way to
  find the good leaves would be to test them one after the other, but this would
  not ensure a good delay bound. Instead, we use the notion of $S^+/S^-$-paths
  to only explore the relevant parts of~$T_\calK$. Following this idea,
  we say that an internal node~$n$ 
  at depth $0 \leq r < \card{\calK}$ of~$T_\calK$ is \emph{good} if there is
  an $S^+/S^-$ path from a vertex of~$\Lambda'$ to a vertex with an outgoing
  $\epsilon$-edge, where $S^+$ and $S^-$ respectively contain the labels of~$\calK$
  that we decided to take and those that we decided not to take when going from
  the root of~$T_\calK$ to~$n$. Note that $S^+, S^-$ is a partition of the
  $r$ first labels of~$\calK$ that uniquely defines~$n$.

  We can now use Lemma~\ref{lem:extendToPath} as an oracle to determine, given
  any node~$n$ of the tree, whether $n$ is good in this sense or not. This
  oracle makes it possible to find the good leaves of~$T_\calK$ 
  efficiently, by starting at the root of~$T_\calK$ and
  doing a depth-first exploration of good nodes of the tree. 
  We build~$T_\calK$ on-the-fly while doing so, to avoid
  materializing irrelevant parts of the tree. The exploration is guaranteed to find all
  good leaves, because the root of the tree is always good, and because the
  ancestors of a good leaf are always good.
  Further, it ensures that we always find one new good leaf after at most
  $O(\card{\calK})$ invocations of Lemma~\ref{lem:extendToPath}, because
  whenever we are at a good node then it must have a good child and therefore, by induction, a good descendant that is a leaf. We
  will find this leaf in our depth-first search with a number of oracle calls that is at
  most linear in the height of~$T_\calK$. Together with the delay bound of
  Lemma~\ref{lem:extendToPath}, this yields the claimed delay bound of
  $O(\card{\Widthfull} \times \Absize^2)$.

  Last, it is clear that whenever we have found a good leaf corresponding to a
  set~$S^+$, then we can compute the new level set $\Lambda''$ that we must
  return together with~$S^+$, with the same delay bound.
  Indeed, we can simply do this by post-processing the set
  of vertices returned by the corresponding invocation of
  Lemma~\ref{lem:extendToPath}.
\end{inlineproof}

\subparagraph*{Memory usage.}
The recursion depth of Algorithm~\ref{alg:main} 
on general sequential VAs is unchanged, and we can still
eliminate tail recursion for the case $\lvar=\emptyset$ as we did in Section~\ref{sec:enum}.

The local space must now include the local space used by the enumeration scheme
of \nextlevel, of which there is an instance running at every level
on the stack. We need to remember our current position in the binary search
tree: assuming that the order of labels is fixed, it suffices to remember the
current positive set $P_n$ plus the last label in the order on~$\calK$ that we
use, with all other labels being implicitly in~$N_n$. This means that we store one label per level (the last
label), plus the positive labels, so their total number in the stack is at
most the total number of markers, i.e., $O(\card{\calV})$. Hence 
the structure of Theorem~\ref{thm:generalVA} has no effect on the memory usage.

The space usage must also include the space used for one call to the
construction of Lemma~\ref{lem:extendToPath}, only one instance of which is
running at every given time. This space usage is clearly in $O(\card{Q}
\times \card{V})$, so this additive term has again no impact on the memory
usage. Hence, the memory usage of our enumeration algorithm is the same as in
Section~\ref{sec:enum}, i.e.,
$O((r+1) \times \Width)$,
or $O((\card{\calV}+1) \times \card{Q})$ in terms of the VA.

\section{Conclusion}
\label{sec:conc}

We have shown that we can efficiently enumerate the mappings of sequential
variable-set automata on input documents, achieving linear-time preprocessing
and constant-delay in data complexity, while ensuring that preprocessing and
delay are polynomial in the input VA even if it is not deterministic.
This result was previously considered as unlikely by~\cite{FlorenzanoRUVV18}, and
it improves on the algorithms in~\cite{FreydenbergerKP18}:
with our algorithm, the delay between outputs does not depend on the input document,
whereas it had a linear dependency on the size of the input document in~\cite{FreydenbergerKP18}.

We will consider different directions for future works.
A first question 
is how to cope with changes to the input document
without recomputing our enumeration index structure from scratch.
This question has been recently studied for other enumeration algorithms, see
e.g.~\cite{amarilli2018enumeration,BerkholzKS17b,BerkholzKS17,BerkholzKS18,losemann2014mso,Niewerth18,NiewerthS18},
but for atomic update operations: insertion, deletion, and relabelings of single
nodes. However, as spanners 
operate on text, we would like 
to use bulk update operations that modify large parts of the text at once:
cut and paste operations, splitting or joining strings,
or appending at the end of a file and removing from the beginning, e.g., in the case of log files with rotation.
It may be possible to show better bounds for these operations than the ones obtained by modifying each individual letter~\cite{NiewerthS18,losemann2014mso}.

A second question is to generalize our result from words to trees, but this is
challenging: the run of a tree automaton is no longer linear in just one
direction, so it is not easy to skip parts of the input similarly to the jump
function of Section~\ref{sec:jump}, or to combine computation that occurs in
different branches.
We believe that these difficulties can be solved and that a similar result can
be shown for trees, but that the resulting algorithm is far more complex: this
point, and the question of updates, are explored in our follow-up
work~\cite{amarilli2019enumeration}.

Finally, it would be interesting to implement our algorithms and evaluate them
on real-world data similarly to the work in~\cite{ArenasMRV16,Morciano17}.
We believe that our techniques are rather simple and easily implementable, at least in the
case of extended VAs.
Moreover, since there are no large hidden constants in any of our constructions, we feel that they might be feasible in practice. Nevertheless, an efficient implementation would of course have to optimize implementation details that we could gloss over in our theoretical analysis since they make no difference in theory but might change practical behavior substantially.

\bibliography{main}
\end{document}